\documentclass[journal,final,onecolumn,letterpaper]{IEEEtran}

\usepackage{times}
\usepackage{epsfig}
\usepackage{amsmath}
\usepackage{amssymb}
\usepackage{amsthm}
\usepackage{color}
\usepackage{colortbl}
\usepackage{float}
\usepackage{cite}
\usepackage{subfig}
\usepackage{verbatim}

\newtheorem{theorem}{Theorem}
\newtheorem{proposition}[theorem]{Proposition}
\newtheorem{lemma}[theorem]{Lemma}
\newtheorem{definition}[theorem]{Definition}
\newtheorem{remark}[theorem]{Remark}

\newcommand{\beq}{\begin{eqnarray}}
\newcommand{\eeq}{\end{eqnarray}}
\newcommand{\n}{\nonumber}

\newcommand{\field}[1]{\mathbb{#1}}

\newcommand{\F}{\field{F}}

\newcommand{\cM}{{\cal M}}
\newcommand{\cC}{{\cal C}}

\newfont{\bbb}{msbm10 scaled 500}

\newfont{\bb}{msbm10 scaled 1100}

\newcommand{\FF}{\mbox{\bb F}}

\newcommand{\av}{{\bf a}}
\newcommand{\bv}{{\bf b}}
\newcommand{\cv}{{\bf c}}
\newcommand{\dv}{{\bf d}}
\newcommand{\ev}{{\bf e}}
\newcommand{\fv}{{\bf f}}

\newcommand{\mv}{{\bf m}}

\newcommand{\rv}{{\bf r}}
\newcommand{\sv}{{\bf s}}

\newcommand{\uv}{{\bf u}}

\newcommand{\vv}{{\bf v}}
\newcommand{\xv}{{\bf x}}
\newcommand{\yv}{{\bf y}}
\newcommand{\zv}{{\bf z}}

\newcommand{\Bm}{{\bf B}}

\newcommand{\Gm}{{\bf G}}

\newcommand{\Id}{{\bf I}}

\newcommand{\Mm}{{\bf M}}

\newcommand{\Vm}{{\bf V}}

\newcommand{\Cc}{{\cal C}}
\newcommand{\Dc}{{\cal D}}
\newcommand{\Ec}{{\cal E}}

\newcommand{\Gc}{{\cal G}}

\newcommand{\Kc}{{\cal K}}

\newcommand{\Mc}{{\cal M}}
\newcommand{\Nc}{{\cal N}}

\newcommand{\Pc}{{\cal P}}

\newcommand{\Rc}{{\cal R}}

\newcommand{\Xc}{{\cal X}}
\newcommand{\Yc}{{\cal Y}}
\newcommand{\Zc}{{\cal Z}}

\definecolor{OXO-emph}{RGB}{153,0,0}

\newcommand\floorb[1]{\left\lfloor #1 \right\rfloor}

\begin{document}

\title{Secure Cooperative Regenerating Codes for Distributed Storage Systems}

\author{O.~Ozan~Koyluoglu,~Ankit~Singh~Rawat,~and~Sriram~Vishwanath
\thanks{O.~O.~Koyluoglu is with the Department of Electrical and Computer Engineering, The University of Arizona, Tucson, AZ 85721, USA (e-mail: ozan@email.arizona.edu).}
\thanks{A.~S.~Rawat and S.~Vishwanath are with the Department of Electrical and Computer Engineering, The University of Texas at Austin, Austin, TX 78712, USA (e-mail: ankitsr@utexas.edu, sriram@austin.utexas.edu).}
}

\maketitle

\begin{abstract}
Regenerating codes enable trading off repair bandwidth for
storage in distributed storage systems (DSS). Due to their distributed nature,
these systems are intrinsically susceptible to attacks, and
they may also be subject to multiple simultaneous node failures. Cooperative regenerating codes allow bandwidth efficient repair of multiple simultaneous node failures. This paper analyzes storage systems 
that employ cooperative regenerating codes that are robust to (passive) eavesdroppers. 
The analysis is divided into two parts, studying both minimum bandwidth and  
minimum storage cooperative regenerating scenarios. First, the secrecy capacity 
for minimum bandwidth cooperative regenerating codes is characterized. 
Second, for minimum storage cooperative regenerating codes, a secure file size 
upper bound and achievability results are provided. These results establish the 
secrecy capacity for the minimum storage scenario for certain special cases. 
In all scenarios, the achievability results correspond to exact repair, 
and secure file size upper bounds are obtained using min-cut analyses 
over a suitable secrecy graph representation of DSS. The 
main achievability argument is based on an appropriate pre-coding of the 
data to eliminate the information leakage to the eavesdropper.
\end{abstract}


\section{Introduction}

\subsection{Background}

Distributed storage systems (DSS) are designed to store data over a distributed network of nodes. DSS have become increasingly important given the growing volumes of data being generated, analyzed, and archived. 
OceanStore \cite{Rhea:Pond03}, Google File System (GFS) \cite{Ghemawat:Google03}, and 
TotalRecall \cite{Bhagwan:Total04} are a few examples of storage systems employed today. Data to be stored is more than doubling every two years, and efficiency in storage and data recovery is particularly critical today. 
The coding schemes employed by DSS are designed to provide efficient storage while ensuring resilience against 
node failures in order to prevent the permanent loss of the data stored 
on the system. In a majority of existing literature, the analysis of DSS  focuses primarily on isolated node failures. In our work, we study a more general scenario of DSS that can suffer from multiple simultaneous node failures. In addition to multiple node failures, DSS systems are also inherently susceptible to adversarial attacks such as one from eavesdroppers aiming to gain access to the stored data. Therefore, it is desirable to have DSS that meet certain security requirements while performing efficient repairs even in the case of multiple simultaneous node failures.

In \cite{Dimakis:Network10}, Dimakis et al. present a class of \emph{regenerating codes},
which efficiently trade off per node storage and repair bandwidth for single node repair.
These codes are designed to possess a property of maximum distance separable (MDS) codes, referred to as \emph{`any $k$ out of $n$'} property, wherein the
entire data can be reconstructed by contacting to any $k$ storage nodes
out of $n$ nodes in the storage system. By utilizing a network coding approach, the notion of 
{\em functional repair} is developed in \cite{Dimakis:Network10}. Here, the original 
failed node may not be replicated exactly after node repair, but can be repaired such that the data stored on the repaired node is {\em functionally} equivalent to that stored on the failed node. On the other hand, {\em exact repair} requires that the regeneration process results in an exact 
replica of the data stored on the failed node. This is essential due to 
the ease of maintenance and other practical requirements such as maintaining 
a code in its systematic form. 

Exact repair may also prove to be advantageous compared 
to functional repair in the presence of eavesdroppers, as the latter 
scheme requires updating the coding coefficients, which in turn may leak additional 
information to eavesdroppers \cite{Pawar:Securing11}. The design of exact regenerating 
codes achieving one of the two ends of the trade off between storage and 
repair bandwidth has recently been investigated by researchers.
In particular, Rashmi et al. \cite{Rashmi:Optimal11} propose codes that are optimal 
for all parameters at the minimum bandwidth regeneration (MBR) 
point. For the minimum storage regeneration (MSR) point, on the other hand, the codes presented in \cite{Rashmi:Optimal11}  have their rate upper bounded by $\frac{1}{2} + \frac{1}{2n}$. In high rate regime (i.e., $\frac{k}{n} > \frac{1}{2}$), the codes at the MSR point have recently been proposed under various restrictions on per node storage $\alpha$  (see e.g., \cite{Papailiopoulos:Repair11,Cadambe:Permutation11,Tamo:Zigzag11,Wang:codes11} and references therein). Some of these codes at the MSR point allow for bandwidth efficient repair of only systematic nodes, e.g., \cite{Cadambe:Permutation11,Tamo:Zigzag11}.

\subsection{Cooperative repair}

As discussed above, DSS can also exhibit multiple simultaneous node failures, and it is desirable that these be repaired simultaneously. It is not uncommon that multiple failures 
occur in DSS, especially for large-scale systems. In addition, some DSS administrators may choose to wait to initiate a repair process after a critical number of failures has occurred (say $t$ of them) in order to render the entire process more efficient and less frequent.
For example, TotalRecall \cite{Bhagwan:Total04} currently executes a node repair process only 
after a certain threshold on the number of failures is reached. In such multiple 
failure scenarios, each new node replacing a failed one can still 
contact $d$ remaining (surviving) nodes to download data for the repair process. 
In addition, replacement nodes, after downloading data from surviving nodes, can 
also exchange data within themselves to complete the repair process. This repair 
process is referred to as {\em cooperative repair} in \cite{Hu:Cooperative10},
which presents network coding techniques to enable such repairs. 
Cooperative repair is shown to be essential as it can help in lowering the 
total repair bandwidth compared to the $t=1$ case. Flexibility of the 
choice on download nodes during the repair process is analyzed in~\cite{Wang:MFR10}. 
\cite{Kermarrec:Repairing11}, focusing on functional repair, shows that under 
the constraint $n=d+t$, deliberately delaying repairs (and thus increasing $t$) does 
not result in gains in terms of MBR/MSR optimality.  \cite{Kermarrec:Repairing11} 
and \cite{Shum:Existence11} utilize a cut-set bound argument and derive cooperative
counterparts of the end points of the trade off curve between repair bandwidth and per node storage for regenerating codes. These two points are named
as the minimum bandwidth cooperative regenerating (MBCR) point and the minimum storage 
cooperative regenerating (MSCR) point (see also~\cite{Oggier:Coding12}).
The work in~\cite{Shum:Existence11} shows the existence of cooperative 
regenerating codes with optimal repair bandwidth. Explicit code constructions 
for exact repair in this setup are presented in~\cite{Shum:Exact11}, for the MBCR point,
and in~\cite{Shum:Cooperative11}, for the MSCR point. These constructions 
are designed for the setting of $d=k$. (See also~\cite{Shum:Cooperative12}.) 
Interference alignment is used in~\cite{LeScouarnec:Exact12} to construct codes to operate 
at the MSCR point. This construction is limited to the case $k=2$ with $d\geq k$, 
and does not generalize to $k\geq 3$ with $d>k$. An explicit construction for the MBCR point, with the restriction that $n=d+t$ for any $t\geq 1$, is presented in~\cite{Jiekak:CROSS12}. Finally, the reference~\cite{Wang:Exact12} presents  codes for the MBCR point for all possible parameter values. Noting the significance of cooperative repair in DSS, regenerating codes that have resilience to eavesdropping attacks will have greater value if 
they also have efficient cooperative repair mechanisms. 

\subsection{Security in distributed storage systems}

The security of systems can be understood in terms of  their resilience to either (or both) active 
or passive attacks~\cite{Goldreich:Foundations04, Delfs:Introduction07}. Active attacks 
include settings where the attacker modifies existing packets or injects new ones into the system, 
whereas passive attacks include eavesdroppers observing the information being 
stored/transmitted. For DSS, cryptographic approaches like private-key 
cryptography are often logistically prohibitive as the secret key distribution between each pair of nodes and 
its renewal are highly challenging, especially for large-scale systems. In addition, 
most cryptographic approaches are typically based on certain hardness results, which, if repudiated, could leave the system
 vulnerable to attacks. On the other hand, information theoretic security, see, e.g., 
\cite{Shannon:Communication49,Wyner:Wire-tap75}, presents secrecy guarantees even with infinite 
computational power at eavesdroppers without requiring the sharing and/or distribution of keys. This approach 
is based on the design of secrecy-achieving coding schemes by taking into account 
the amount of information leaked to eavesdroppers, and can offer new solutions 
to security challenges in DSS. In its simplest form, the security 
can be achieved with the one-time pad scheme~\cite{Vernam:Cipher26}, which claims the security
of the data in a ciphertext, which is obtained by XOR of the data and a uniform key. This approach is of significant
value to DSS. For example, consider a system storing the key at a node, and the ciphertext 
at another node. Then, the eavesdropper can not obtain any information regarding the data by observing one of these two nodes, whereas the data collector can contact to both nodes and decipher the data.

The problem of designing secure DSS against eavesdropping attacks has recently been studied by Pawar et al. \cite{Pawar:Securing11}, where the authors consider 
a passive eavesdropper model by allowing an eavesdropper to observe the data stored on any $\ell$ $(<k)$ 
storage nodes in a DSS employing an MBR code. The proposed schemes are designed
for the `bandwidth limited regime', and shown to achieve an upper bound
on the secure file size, establishing its optimality. Shah et al.~\cite{Shah:Information11} 
consider the design of secure MSR codes. They point out that the eavesdropper model 
for an MSR code should be extended compared to that of an MBR code. The underlying
reason is that at the MSR point of operation, an eavesdropper may obtain additional information by 
observing the data downloaded during a node repair (as compared to just observing the stored content). Thus, at the MSR point, the eavesdropper is modeled with
a pair ($\ell_1,\ell_2$) with $\ell_1+\ell_2<k$ specifying the eavesdropper that
has knowledge of the content of an $\ell_1$ number of nodes, and, in addition, 
has knowledge of the downloaded information (and hence also the stored content) 
of an $\ell_2$ number of nodes. We note that, as the amount of the data downloaded during a node repair and per node storage for MBR codes are the same, the two notions are 
different only at the MSR point. Considering such an eavesdropper 
model, Shah et al. present coding schemes utilizing product matrix 
codes \cite{Rashmi:Optimal11}, and show that the bound on secrecy capacity 
in \cite{Pawar:Securing11} at the MBR point is achievable. They further use product matrix based 
codes for the MSR point as well, and show that the bound in \cite{Pawar:Securing11} is achievable
only when $\ell_2 = 0$. More recently, based on appropriate maximum rank distance pre-coding of Zigzag codes~\cite{Tamo:Zigzag11} and their variants \cite{Wang:codes11}, secure codes for the MSR point are proposed in~\cite{Rawat:Optimal12}. This construction is shown to achieve the secrecy capacity for a class of systems where only the downloads of the systematic nodes are observed by the eavesdropper in~\cite{Goparaju:Data13} for $d=n-1$ for any $(\ell_1,\ell_2)$. 
Besides this classical MBR/MSR setting, the 
security aspects of locally repairable codes (see, e.g., \cite{Oggier:Homomorphic11, Oggier:Self11,Gopalan:locality11,Huang:Pyramid07,Prakash:Optimal12,Papailiopoulos:Locally12}) are studied in~\cite{Rawat:Optimal12}; and, security of DSS against active eavesdroppers is investigated in~\cite{Oggier:Byzantine11,
Rashmi:Regenerating12, Silberstein:Error12, Han13}.

\subsection{Contributions and organization}

In this paper, we analyze and design secure and cooperative regenerating codes for 
DSS. In terms of security requirements, we utilize a 
passive and colluding eavesdropper model as presented in \cite{Shah:Information11}. In this model, 
during the entire life span of the DSS, the eavesdropper can gain access 
to data stored on an $\ell_1$ number of nodes, and, in addition, it observes 
both the stored content and the data downloaded (for repair) on an 
additional $\ell_2$ number of nodes. Given this eavesdropper model, we
focus on the problem of designing secure cooperative regenerating codes in the 
context of DSS that perform multiple simultaneous node repairs. 
This scenario generalizes the single node repair setting considered in 
earlier works to multiple node failures. In this paper, we establish the secrecy capacity for the MBCR point, and propose some secure codes for the MSCR point that are optimal for some special cases. In all scenarios, the achievability results allow for exact repair. The main secrecy achievability coding argument in this paper is obtained by utilizing a secret pre-coding (randomization) scheme to obtain secure coding schemes for DSS. In some cases, this pre-coding 
is established simply with the one-time pad scheme, and in others {\em maximum rank distance} (MRD) codes are utilized similar to the classical work of~\cite{Shamir:How79}. We remark that utilization of such pre-coding mechanisms is critical in showing the security of the proposed schemes. In particular, our security proofs are based on an oracle-type proof, where the eavesdropper is asked to decode the random symbols given the secure data symbols in addition to its observations. We design the pre-coding mechanisms to allow for such a security analysis. For example, when we utilize MRD pre-coding, we show that the eavesdropper has enough number of evaluations of an underlying polynomial (used in MRD pre-coding) at linearly independent points to resolve for the random symbols.
We summarize our contributions in the following.
\begin{itemize}
\item We present an upper bound on the secrecy capacity for MBCR codes. This bound follows from the information theoretic analysis of counting the \emph{secure flow} for a particular repair instance in DSS employing an MBCR code.
\item We present a secure MBCR code for $n=d+t$ by employing a maximum rank distance pre-coding of the codes proposed in~\cite{Jiekak:CROSS12}. By comparing the secure file size of this code with the upper bound obtained, we characterize the secrecy capacity for MBCR codes for $n=d+t$ (for any $\ell_1$).
\item For $n>d+t$, we present secrecy capacity achieving codes (for any $\ell_1$) by utilizing the bivariate polynomial proposed in~\cite{Wang:Exact12}. In particular, our construction is based on randomizing some appropriate coefficients of the underlying bivariate polynomial.
\item For MSCR codes, we obtain an upper bound on the secrecy capacity against a passive eavesdropper that takes into account the amount of information leaked to the download-observing eavesdroppers.
\item We present a secure MSCR code for $k=t=2$ based on the code proposed in~\cite{LeScouarnec:Exact12}. In particular, we place random symbols on some nodes in DSS in addition to utilizing an appropriate one-time pad scheme for securing the data stored on other nodes. This construction is shown to achieve optimal secure file size. (Note that, for the case where $k=t=2$, under the restriction that $\ell_1 + \ell_2 < k$ a non-trivial eavesdropper can only be associated with $(\ell_1, \ell_2) = (1, 0)$ or $(\ell_1,\ell_2) = (0, 1)$.)
\item Finally, we construct secure MSCR codes when $d=k$ and characterize achievable secure file size using such codes for any $(\ell_1,\ell_2)$. This construction is based on maximum rank distance pre-coding of the codes proposed in~\cite{Shum:Cooperative11}. We show that this construction achieves the secrecy capacity of MSCR codes for a restrictive eavesdropper model specified by $(\ell_1,\ell_2)$ with $\ell_2 \leq 1$.
\end{itemize}

The rest of the paper is organized as follows. In Section II, we 
provide the general system model together with some preliminary results 
utilized throughout the text. Section III provides the analysis of 
secure MBCR codes, and Section IV is devoted to the secure MSCR 
codes. The paper is concluded in Section V. Some of the results and proofs are relegated to appendices to enhance the flow 
of the paper.


\section{System Model and Preliminaries}
\label{sec:sys_model}

Consider a DSS consisting of $n$ live nodes (at a given time)
and a file $\fv$ of size $\mathcal{M}$ over a finite field $\mathbb{F}$ that 
needs to be stored on the DSS\footnote{The size of $\F$ is specified later in the context of specific coding schemes.}. The file $\fv$ is encoded into $n$ data blocks $\xv = (\mathbf{x}_1,\ldots, \mathbf{x}_n)$, each of length $\alpha$ over 
$\mathbb{F}$ with $\alpha \geq \frac{\mathcal{M}}{k}$. 
Given the codeword $\xv$, node $i$ in an $n$-node DSS stores encoded block 
$\xv_i$. (In this paper, we use $\xv_i$ to represent both block $\xv_i$ 
and a storage node storing this encoded block interchangeably.) Motivated by 
the MDS property of the codes that are traditionally proposed to store data 
in centralized storage systems \cite{Patterson:case88, Blaum:EVENODD95, Blaum:MDS96}, the works 
on regenerating codes focus on storage schemes that have `any $k$ out of $n$'
property, i.e., the contents of any $k$ nodes suffice to recover the
file. We focus on codes with this property\footnote{Note that having `any $k$ out of $n$' property does not necessarily imply that the code is an MDS code. A vector code $\Cc$ defined over $\F_q$ with symbols in $\F_q^\alpha$ is said to be MDS if $\alpha|\log_q |\Cc|$ and $d_\textrm{min}=n - \frac{\log_{q}{|\cC|}}{\alpha} + 1$.}.

We use the following notation throughout the text.
We usually stick with the notation of having vectors 
denoted by lower-case bold letters; sets and subspaces are denoted 
by calligraphic fonts. For $a<b$, $[a:b]$ represents the set of 
numbers $\{a,a+1,\cdots, b\}$.  Similarly, $a_{{i_1:i_2}, {j_1:j_2}}$ denotes the set $\{a_{i_1, j_1},\ldots, a_{i_2, j_1},\ldots, a_{i_1,j_2},\ldots, a_{i_2, j_2}\}$. In addition, $a_{i_1:i_2, j}$ and $a_{i, j_1:j_2}$ denote $\{a_{i_1, j},\ldots, a_{i_2, j}\}$ and $\{a_{i, j_1},\ldots, a_{i, j_2}\}$, respectively.
The symbols stored at node $i$ is represented by the vector $\sv_i$, 
the symbols transmitted from node $i$ to node $j$ is denoted as $\dv_{i,j}$, 
and the set $\dv_j$ is used to denote all of the symbols downloaded at node $j$. 
DSS is initialized with the $n$ nodes containing encoded symbols, i.e., 
$\sv_i=\xv_i$ for $i=1, \cdots, n$.


\subsection{Cooperative repair in DSS}

In most of the studies on DSS, exact repair for regenerating codes is analyzed in the context of single node failure. However, it is not uncommon to see simultaneous multiple node
failures in storage networks, especially for large ones.
The basic setup involves the simultaneous repair of $t$ (greater
than one) failed nodes. After the failure of $t$ storage 
nodes, the same number of newcomer nodes are introduced to the system. Each newcomer node
connects to $d$ arbitrary live storage nodes and downloads $\beta$ symbols from each of 
these nodes. In addition, utilizing a cooperative approach, each newcomer node
also contacts other newcomer nodes involved in node repair process and downloads $\beta'$ symbols from 
each of these nodes. Hence, the total repair cost is given by
\begin{equation*}
\gamma=d\beta + (t-1) \beta'.
\end{equation*}
Each newcomer node, to repair the $i$th node of the original
network, uses these $d\beta+(t-1)\beta'$ number of 
downloaded symbols to regenerate $\alpha$ symbols, $\mathbf{x}_i$, and stores 
these symbols. The $t$ nodes simultaneously repaired in a cooperative manner constitute a repair group. 

This exact repair process preserves the `$k$ out of $n$ property', i.e., 
data stored on any $k$ nodes (potentially including the nodes that are 
repaired) allows the original file $\fv$ to be reconstructed.
See Fig.~\ref{fig:FlowGraph1}.

We remark that, as also argued in~\cite{Wang:Exact12},
$d\geq k$ can be assumed without loss of generality.
(Earlier papers on the subject assumed $d\geq k$ for simplicity.
See, e.g.,~\cite{Shum:Cooperative11,Shum:Exact11,LeScouarnec:Exact12,
Jiekak:CROSS12,Shum:Cooperative12}.)
Remarkably, if $d<k$, a data collector can reconstruct
the whole file by contacting only $d$ nodes, as from
these nodes the other nodes can be repaired in groups
of size $t$. Thus, any $(n,k,d)$ code with $d<k$ can be
reduced to $(n,k'=d,d)$ code. Therefore, without loss
of generality, we will assume $d\geq k$.

\begin{figure}[t]
\centering
\includegraphics[width=0.5\columnwidth]{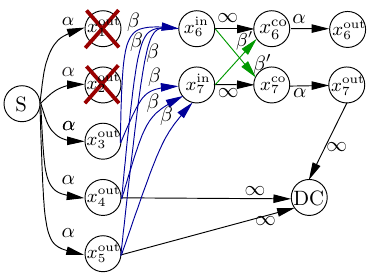}
\caption{Information flow graph of DSS implementing cooperative repair.
In this representative example, we have $n=5$, $d=k=3$, and $t=2$. Accordingly,
after a failure of two nodes, namely node $1$ and node $2$, the system
cooperatively repairs these two nodes as node $6$ and node $7$. Downloads
from live nodes (blue) and from cooperative repair pairs (green) are shown.
Due to exact repair, the network will repair the nodes to satisfy 
$x^{\rm out}_6=x^{\rm out}_1$ and $x^{\rm out}_7=x^{\rm out}_2$.}
\label{fig:FlowGraph1}
\end{figure}


\subsection{Information flow graph} 

In their seminal work \cite{Dimakis:Network10}, Dimakis et al. model the operation of 
DSS using a multicasting problem over an information flow graph. 
(See Figs.~\ref{fig:FlowGraph1} and \ref{fig:FlowGraph2} 
for the flow graph in the cooperative setting.) An information flow graph representation of a DSS consists of three types of nodes:
\begin{itemize}
\item Source node ($S$): Source node contains $\mathcal{M}$ symbols 
long original file $\fv$. The source node is connected to $n$ nodes.
\item Storage nodes ($(x^{\rm in}_i,x^{\rm co}_i,x^{\rm out}_i)$): 
In an information flow graph associated with cooperative repairs, we represent each node with a combination 
of three sub-nodes: $x^{\rm in}$, $x^{\rm co}$, and $x^{\rm out}$. 
Here, $x^{\rm in}$ is the sub-node having the connections from the 
live nodes, $x^{\rm co}$ is the sub-node having the connections 
from the nodes under repair in the same repair group, and $x^{\rm out}$ 
is the storage sub-node which represents the content stored on the corresponding node in DSS. $x^{\rm out}$  is contacted by a data collector or other nodes during node repairs. $x^{\rm in}$ is connected 
to $x^{\rm co}$ with a link of infinite capacity, $x^{\rm co}$ is 
connected to $x^{\rm out}$ with a link of capacity $\alpha$. We represent 
cuts using a notation with bars as in $(x^{\rm in},x^{\rm co}|x^{\rm out})$, 
meaning the cut is passing through the link between $x^{\rm co}$ and 
$x^{\rm out}$. (See Fig.~\ref{fig:FlowGraph2}.) The nodes on the right hand 
side of the cuts belong to data collector side, represented by the set $\Dc$, 
whereas the nodes belonging to the left hand side of the cuts belong to
$\Dc^c$, the source side.
For a newcomer node $i$, $x^{\rm{in}}_i$ is connected to $x^{\rm{out}}$ 
sub-nodes of $d$ live nodes with links of capacity $\beta$ symbols each, 
representing the data downloaded during node repair. $x^{\rm co}_i$ sub node associated with this newcomer node
also connects to $x^{\rm{in}}$ sub-nodes of $(t-1)$ remaining nodes that are being repaired in the 
same group, and each link of these connections has a capacity of $\beta'$.
\item Data collector node(s) ($\rm{DC}$): Each data collector contacts 
$x^{\rm{out}}$ sub-node of $k$ live nodes with edges of $\infty$-link capacity.
\end{itemize}


\subsection{MBCR and MSCR points}
\label{sec:MBCRMSCR}

With the aforementioned values of capacities of various edges in the 
information flow graph, the DSS is said to employ an $(n,k,d,\alpha,\beta,\beta')$ 
code. For a given graph $\Gc$ and data collectors $\rm{DC}_i$, the file 
size that can be stored in such a DSS can be bounded using the max 
flow-min cut theorem for multicasting utilized in 
network coding~\cite{Ahlswede:Network00,Ho:random06}.
\begin{lemma}[Max-flow min-cut theorem for multicasting~\cite{Ahlswede:Network00,Ho:random06,Dimakis:Network10}]
$$\Mc \leq \min_{\Gc} \min_{\rm{DC}_i} \rm{max flow}(S \to \rm{DC}_i,\Gc),$$
where $\rm{flow}(S \to \rm{DC}_i,\Gc)$ represents the flow from the source 
node $S$ to data collector $\rm{DC}_i$ over the graph $\Gc$.
\end{lemma}
Therefore, e.g., for the graph in Fig.~\ref{fig:FlowGraph2}, 
$\mathcal{M}$ symbol long file can be delivered to a data collector $\rm{DC}$, 
only if the min cut is at least $\mathcal{M}$. 

Dimakis et al. \cite{Dimakis:Network10} obtain the following bound (for $t=1$ case) by considering $k$ successive node failures
and evaluating the min-cut over all possible graphs.
\begin{align}
\mathcal{M} \leq \sum_{i=0}^{k-1}\min\{\alpha, (d - i)\beta\}
\label{eq:dimakis_thm}
\end{align}
We emphasize that the min-cut for this ($t=1$) case is given by the scenario
where $k$ successively repaired nodes are connected to $\rm{DC}$, and for each successive repair, the repaired node $i+1$ also connects to
$i$ number of previously repaired nodes. Hence, for each $\rm{DC}$-connected node, 
its contribution to the value of a cut from S to {\rm DC} is equal to $(d-i)\beta$ if the cut through the node is of type $(|x^{\rm in},x^{\rm out})$,
and is equal to $\alpha$ if the cut separates two sub-nodes, i.e., the cut through the node is of type $(x^{\rm in}|x^{\rm out})$.
(Note that, $x^{\rm co}$ does not appear here as the model considered in~\cite{Dimakis:Network10}
does not involve cooperative repair.) The codes that attain the bound
in (\ref{eq:dimakis_thm}) are named as {regenerating codes} \cite{Dimakis:Network10}.

\begin{figure}[t]
\centering
\includegraphics[width=0.5\columnwidth]{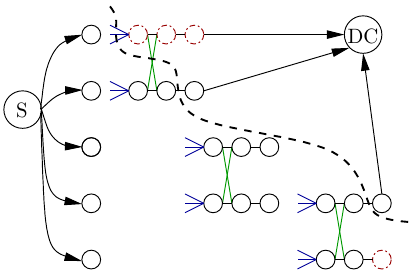}
\caption{Information flow graph of DSS implementing cooperative repair
under security constraints. In this representative example, we have 
$n=5$, $d=k=3$, and $t=2$. Multiple repair stages and a cut, represented by
dashed line, through the nodes connected to the $\rm{DC}$ are shown. 
The figure has different cut types: 
The first repaired node has a cut of type $(|x^{\rm in},x^{\rm co},x^{\rm out})$
and the second has a cut of type $(x^{\rm in},x^{\rm co}|x^{\rm out})$.
Nodes that are being eavesdropped are indicated with dashed-dotted lines. 
Here, both the content and the downloads
of the first repaired node is observed by the eavesdropper ($\ell_2=1$),
and only the content of the last repaired node is observed by the eavesdropper
($\ell_1=1$). Accordingly, eavesdropper has observations of $d\beta+(t-1)\beta'$ 
downloaded symbols from the first repaired node, and has $\alpha$ number of symbols
from the last repaired node.}
\label{fig:FlowGraph2}
\end{figure}

Analysis of the cut-set bounds for cooperative regenerating codes
are provided in~\cite{Kermarrec:Repairing11,Shum:Existence11}.
(See also the arguments given in~\cite{Hu:Cooperative10,Oggier:Coding12}.
Here, we follow the notations of~\cite{Kermarrec:Repairing11,Oggier:Coding12}.)
Consider a scenario where groups of nodes (each group having $t$ nodes) are consecutively repaired in DSS. Let us enumerate the groups that are consecutively repaired as $i=0,\cdots,\mu-1$. There are in total $\mu t$ number of nodes in this repair process, and consider that $\rm{DC}$ contacts $k$ out of these $\mu t$ nodes. Let us denote the number of nodes in repair group $i$ that are contacted by the $\rm{DC}$ as $n_i$ such that
$n_i \in [0:t],~\text{for all}~i \in  \{0,1,\ldots,\mu-1\},$
and
$\sum\limits_{i=0}^{\mu-1} n_i = k.$

The cut-set bound for this scenario is given by the following.
\begin{align}\label{eq:CoopCutsetBound}
&\Mc \leq \sum\limits_{i=0}^{\mu-1}
n_i \min \Bigg\{
\alpha, \left(d-\sum\limits_{j=0}^{i-1}n_j\right)\beta
+
\left(t-n_i\right)\beta'
\Bigg\}.
\end{align}
Similar to the $t=1$ case described above, the cut of type
$(x^{\rm in},x^{\rm co}|x^{\rm out})$ has a value of $\alpha$.
The cut of type $(|x^{\rm in},x^{\rm co},x^{\rm out})$, on the other
hand, has a value of $\left(t-n_i\right)\beta'$ due to the links coming
from the nodes under repair that are not connected to $\rm{DC}$ and additional value of $(d-\sum\limits_{j=0}^{i-1}n_j)\beta$
due to the connections to the previously repaired live nodes that are not contacted by $\rm{DC}$.
(Here, we again subtract the values of the flows from the nodes 
already belonging to the data collector side, $\Dc$.)
The cut of type $(x^{\rm in}|x^{\rm co},x^{\rm out})$ has value
of $\infty$ and hence, does not appear in the min-cut.

Note that, given a file size $\mathcal{M}$, there is an inherent 
trade off between storage per node $\alpha$ and {\em repair bandwidth} 
$\gamma \triangleq d\beta+(t-1)\beta'$. This trade off, for the cooperative 
setting, can be established using a similar analyses leading to the MBR/MSR 
points from the equation \eqref{eq:dimakis_thm}.
Two classes of codes that achieve two extreme points of this trade off 
are named as 
{\em minimum bandwidth cooperative regenerating (MBCR)} codes and 
{\em minimum storage cooperative regenerating (MSCR)} codes. The former 
is obtained by 
first finding the minimum possible $\gamma$ and then finding the 
minimum $\alpha$ satisfying \eqref{eq:CoopCutsetBound}.
This point is given by the following.
\begin{align}
&\alpha_{\rm{MBCR}}=\frac{\Mc}{k}\frac{2d+t-1}{2d+t-k}, \quad
\gamma_{\rm{MBCR}}=\alpha_{\rm{MBCR}}, \nonumber\\
&\beta_{\rm{MBCR}}=\frac{\Mc}{k}\frac{2}{2d+t-k}, \quad
\beta'_{\rm{MBCR}}=\frac{\Mc}{k}\frac{1}{2d+t-k} \label{eq:mbcr_point}
\end{align}

The MSCR point, on the other hand, is obtained by 
first choosing a minimum storage per node (i.e., $\alpha=\Mc /k$), and then 
minimizing $\gamma$ (via choosing minimum possible $\beta$-$\beta'$ pair) 
satisfying the min cut \eqref{eq:CoopCutsetBound}.

\begin{align}
&\alpha_{\rm{MSCR}}=\frac{\mathcal{M}}{k}, \quad
\gamma_{\rm{MSCR}}=\frac{\Mc}{k}\frac{d+t-1}{d+t-k}, \nonumber\\
&\beta_{\rm{MSCR}}=\frac{\mathcal{M}}{k}\frac{1}{d+t-k}, \quad
\beta'_{\rm{MSCR}}=\frac{\mathcal{M}}{k}\frac{1}{d+t-k}\label{eq:mscr_point}
\end{align}

We depict these two trade off points, which are directly computable 
from \eqref{eq:CoopCutsetBound}, in Fig.~\ref{fig:MBCR-MSCR}. Note that, when $t=1$, these two 
points correspond to the MBR/MSR points characterized in~\cite{Dimakis:Network10}.
(We refer reader to~\cite{Kermarrec:Repairing11,Shum:Existence11} for a detailed derivation of these two points. See also~\cite{Oggier:Coding12} for an analysis for the simplified case of when $t|k$, i.e., the number of groups satisfies $\mu=k/t$.) We note that, in the next section, we consider secure file size upper bound using similar min cut arguments in the presence of eavesdroppers.

\begin{figure}[t]
\centering
\includegraphics[width=0.6\columnwidth]{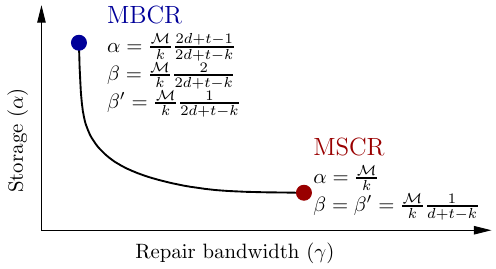}
\caption{Storage vs. repair bandwidth trade off for cooperative regenerating codes.
The repair bandwidth is given by $\gamma=d\beta+(t-1)\beta'$.}
\label{fig:MBCR-MSCR}
\end{figure}


\subsection{Eavesdropper model and Security in DSS}

We consider an $(\ell_1,\ell_2)$ eavesdropper, which has access to the stored data of any $\ell_1$ number of nodes, and additionally has access to both the stored and downloaded data of any $\ell_2$ number of  nodes. Therefore, $(\ell_1,\ell_2)$ eavesdropper has access to 
$x_i^{\rm{out}}$ for $i\in\Ec_1$ and $x_j^{\rm{in}},x_j^{\rm{co}},x_j^{\rm{out}}$ for $j\in\Ec_2$ for some $\Ec_1,\Ec_2$ such that $\Ec_1,\Ec_2\subset[1:n]$, $|\Ec_1|=\ell_1$, $|\Ec_2|=\ell_2$, and $\Ec_1\cap\Ec_2=\emptyset$. (See Fig.~\ref{fig:FlowGraph2}.)
This is the eavesdropper model defined in \cite{Shah:Information11} (adapted here to the 
cooperative repair setting), which generalizes the eavesdropper model 
considered in \cite{Pawar:Securing11}. The eavesdropper is assumed to know the 
coding scheme employed by the DSS. At the MBCR point, a newcomer downloads 
$\alpha_{\rm{MBCR}}=\gamma_{\rm{MBCR}}$ amount of data. Thus, allowing an eavesdropper to access the data downloaded during a node repair besides the content stored on the node does not provide the eavesdropper with any additional information. However, at the MSCR point, repair bandwidth is strictly greater than the per node storage $\alpha_{\rm{MSCR}}$ and an eavesdropper potentially gains more information 
if in addition to content stored on a node it also has access to the data downloaded during repair of the same node. 
We summarize the eavesdropper model together with the definition of 
achievability of a secure file size in the following.
\begin{definition}[Security against an $(\ell_1,\ell_2)$ eavesdropper]\label{def:Eavesdropper}
A DSS is said to achieve a secure file size of $\Mc^s$ 
against an $(\ell_1,\ell_2)$ eavesdropper, if for any sets $\Ec_1$ and $\Ec_2$ 
of size $\ell_1$ and $\ell_2$, respectively, $I(\fv^s;\ev)=0$. Here, $\fv^s$ 
is the secure file of size $\Mc^s$, which is first encoded to file $\fv$ of 
size $\Mc$ before storing it on DSS, and 
$\ev$ is the eavesdropper observation vector given by 
$\ev\triangleq\{x_i^{\rm{out}},x_j^{\rm{in}},x_j^{\rm{co}},x_j^{\rm{out}}: 
i\in\Ec_1, j\in\Ec_2\}$. We use $I(\cdot; \cdot)$ to denote mutual information.
\end{definition}

We remark that, as it will be clear from the following sections, when a file 
$\fv$ (of size $\Mc$) which contains a secure file (of size $\Mc^s$) is stored on a DSS, the remaining $\Mc-\Mc^s$ symbols of $\fv$ can be utilized as an additional data, which does not have security constraints. Yet, noting the possibility of 
storing this insecure data, we will refer to this uniformly distributed part 
as the random data, which is utilized to achieve security. Thus, we consider files of form $\fv=(\fv^s,\rv)$, where $\rv$ can represent random/insecure data. (We remark that DSS properties such as repair bandwidth and per node storage at the MBCR/MSCR point are defined over the file size $\Mc$.)

Here, we note the following lemma, which will be used 
in the following parts of the sequel.
\begin{lemma}[Secrecy Lemma~\cite{Wyner:Wire-tap75,Shah:Information11}]\label{thm:SecrecyLemma}
Consider a system with secure file symbols $\fv^s$,
random symbols $\rv$ (independent of $\fv^s$),
and an eavesdropper with observations
given by $\ev$. If $H(\ev)\leq H(\rv)$ and
$H(\rv|\fv^s,\ev)=0$, then $I(\fv^s;\ev)=0$.
\end{lemma}
\begin{proof}
See Appendix~\ref{app:SecrecyLemma}.
\end{proof}

Finally, we provide a lemma summarizing how the cut values (from source to data collector $\rm{DC}$) in the flow graph can be computed while some of the nodes are eavesdropped in the network.
\begin{lemma}{[File size upper bound under secrecy constraints~\cite{Rawat:Optimal12}]}\label{thm:SecureCutLemma}
Consider a secure file $\fv^{s}$ of size $\Mc^s$, i.e., $\Mc^s=H(\fv^{s})$. Consider an $(\ell_1,\ell_2)$ eavesdropper observing the nodes in the sets $\Ec_1,\Ec_2$ for some $\Ec_1,\Ec_2\subset \{1,\cdots,n\}$ as defined in Definition~\ref{def:Eavesdropper}. Consider also that the data collector $\rm{DC}$ contacts to nodes in the set $\Kc=\Ec_1'\cup\Ec_2'\cup \Rc$ for some $\Ec_1'\subseteq\Ec_1$, $\Ec_2'\subseteq\Ec_2$, and $|\Kc|=k$. (We assume $\ell_1+\ell_2< k$, as otherwise the secrecy capacity of the network is zero.) Enumerating the nodes in $\Kc$ as $\{1,\cdots, k\}$, we have
\begin{equation}\label{eq:SecureCutBound}
\Mc^s \leq \sum\limits_{j=1}^k H(\sv_j|\sv_1,\cdots,\sv_{j-1},\sv_{\Ec_1'},\dv_{\Ec_2'}).
\end{equation}
\end{lemma}
\begin{proof}
The proof is summarized in Appendix~\ref{app:SecureFileSizeBound}. (See also~\cite{Rawat:Optimal12}.)
\end{proof}	
We use \eqref{eq:SecureCutBound} (in some cases with a loose bound on the conditional entropy term) in the the following sections to obtain bounds on the secure file size. We remark that in addition to discounting for the previously contacted nodes $\sv_1,\cdots,\sv_{i-1}$, this bound also discounts for the \emph{leakage} to the eavesdropper, by taking into account the leakage to the set $\{\sv_{\Ec_1'},\dv_{\Ec_2'}\}$. Here, the minimum bound on the file size occurs when $\Ec_1'=\Ec_1$ and $\Ec_2'=\Ec_2$ in the lemma above. Hence, in order to obtain tighter upper bounds on secure file size, we consider only the scenarios for which the data collector connects to all the nodes being eavesdropped.


\subsection{Maximum rank distance (MRD) codes via Gabidulin construction}
\label{sec:MRD}

In some of the following sections of the sequel, we consider maximum rank distance (MRD) pre-coding of the data at hand before storing it on DSS by using cooperative regenerating codes. Here, we introduce the Gabidulin construction of MRD codes \cite{Gabidulin:Theory85, Delsarte:Bilinear78, 
Roth:Maximum91, MacWilliams:Theory77}.

First, we introduce some notation. 
In vector representation, the norm of a vector $\vv\in\F_{q^m}^N$ is the column rank of $\vv$ over the base field $\F_q$, denoted by $Rk(\vv|\F_q)$. (This is the maximum number of linearly independent coordinates of $\vv$ over the base field $\F_q$, for a given 
basis of $\F_{q^m}$ over $\F_q$. A basis also establishes an isomorphism between $N$-length vectors, in $\F_{q^m}^N$, to $m\times N$ matrices, in $\F_q^{m\times N}$. Then, $Rk(\vv|\F_q)=\textrm{rank}(\Vm)$, where $\Vm$ is the corresponding matrix of $\vv$ for the given basis.) Rank distance between two vectors $\vv_1, \vv_2 \in \F_{q^m}^N$ is defined as $d(\vv_1,\vv_2)
=Rk(\vv_1-\vv_2|\F_q)$. (In matrix representation, this is equivalent
to the rank of the difference of the two corresponding matrices of the 
vectors, i.e., $\textrm{rank}(\Vm_1-\Vm_2)$.) Codes utilizing this distance metric are referred to as rank metric codes.

An $[N,K,D]$ rank metric code over the extension field
$\F_{q^m}$ achieving the maximum rank distance $D=N-K+1$ (for $m\geq N$)
can be constructed with the following linearized polynomial.
(This is referred to as the Gabidulin construction of MRD codes, or
Gabidulin codes \cite{Gabidulin:Theory85, Delsarte:Bilinear78, 
Roth:Maximum91, MacWilliams:Theory77}.)
\begin{eqnarray}\label{eq:MRDPoly}
f(g)=\sum\limits_{i=0}^{K-1} u_i g^{[i]},
\end{eqnarray}
where $[i]\triangleq q^i$, $u_i\in\F_{q^m}$, and $g$ is an indeterminate which takes value in $\F_{q^m}$. Given
$N$ linearly independent (over $\F_q$) points, $\{g_1, g_2,\ldots, g_N\} \subset \F_{q^m}$, the codeword $\cv = (c_1, c_2,\ldots, c_N)$ for a given set of $K$ message (information) symbols $\uv = (u_0, u_2,\ldots, u_{K-1})  \in \F^{K}_{q^m}$
is obtained as 
$c_j=f(g_j)=\sum\limits_{i=0}^{K-1} u_i g_j^{[i]}$
for $j=[1:N]$.
With generator matrix representation, we have $\cv=\uv\Gm$,
where $\Gm=[g_1, \cdots, g_N; \cdots; g_1^{[K-1]}, \cdots, g_N^{[K-1]}]$.

We remark that a linearized polynomial $f(\cdot)$ satisfies $f(a_1g_1+a_2g_2)=a_1f(g_1)+a_2f(g_2)$, for
any $a_1,a_2\in\F_q$ and $g_1,g_2\in\F_{q^m}$. This property is utilized in our code constructions.


\section{Secure MBCR codes}
\label{sec:mbcr}

In this section, we study secure minimum bandwidth cooperative regenerating codes.
We first present an upper bound on the secure file size that can be supported
by an MBCR code. Then, we present exact repair coding schemes achieving the
derived bound. In addition, we analyze how the cooperation affects 
the penalty paid in securing storage systems.

\subsection{Upper bound on secure file size of MBCR codes}
\label{subsec:mbcr_bound}

We have the following result for upper bound on secure file size for a DSS employing MBCR codes.

\begin{proposition}\label{thm:MBCRbound}
Cooperative regenerating codes operating at the MBCR
point with a secure file size of $\mathcal{M}^{s}$ satisfy
\begin{eqnarray}
\mathcal{M}^{s} &\leq& k(2d-k+t)\beta'-\ell_1(2d-\ell_1+t)\beta' \nonumber\\
&=& (k-\ell_1)(2d+t-k-\ell_1)\beta' \label{eq:MBCRbound},
\end{eqnarray}
and the MBCR point is given by $\beta=2\beta'$,
$\alpha=\gamma=(2d+t-1)\beta'$ for a file size of $\mathcal{M}=k(2d-k+t)\beta'$.
\end{proposition}

\begin{proof}
We consider the scenario where $\mu$ groups of nodes (each group having $t$ nodes) are consecutively repaired in DSS as introduced in Section~\ref{sec:MBCRMSCR}. Accordingly, the data collector $\rm{DC}$ contacts $n_i\in[0:t]$ nodes in the $i${th} repair group such that $\sum\limits_{i=0}^{\mu-1} n_i = k$. Without loss of generality we index these nodes as $\Kc = \{1,\ldots, k\}$. We consider two types of cuts in each group contacted by the $\rm{DC}$: $m_i$ number of nodes have the first cut type $(x^{\rm in}, x^{\rm co}|x^{\rm out})$, and $n_i-m_i$ number of nodes have the second cut type $(|x^{\rm in}, x^{\rm co}, x^{\rm out})$, $0\leq i\leq \mu-1$\footnote{Note that the cuts of the form $(x^{\rm in}, x^{\rm co}|x^{\rm out})$ give a cut value of $\alpha$ as opposed to $(x^{\rm in}|x^{\rm co}, x^{\rm out})$, which has cut value larger than $\alpha$. Since we are interested in the cuts of smaller size, we do not consider the cuts $(x^{\rm in}|x^{\rm co}, x^{\rm out})$.}.

We consider $\ell_1$ number of colluding eavesdroppers, each
observing the contents of different nodes. Note that, for MBCR point analysis, we can consider $\ell_2=0$ without loss of generality, as the amount of the data a particular node stores is equal to the amount of the data it downloads during its repair.
We denote the number of eavesdroppers on the nodes in the
first cut type as $l_1^{i,1}$, $0\leq i\leq \mu-1$, and the number of
eavesdroppers on the nodes in the second cut type as $l_1^{i,2}$, $0\leq i\leq \mu-1$, such that
\begin{align*}
l_1^{i,1} \leq m_i,\\
l_1^{i,2} \leq n_i-m_i,
\end{align*}
and
\begin{align*}
\sum\limits_{i=0}^{\mu-1} l_1^{i} \leq \ell_1,
\end{align*}
where $l_1^{i} = l_1^{i, 1} + l_1^{i, 2}$. We consider the repair scenario represented in Fig.~\ref{fig:MBCRCut} and utilize Lemma~\ref{thm:SecureCutLemma} to obtain the desired bound.
Here, for repair group $i$, due to the eavesdroppers,
the nodes that belong to the first type can only add the
value of $(m_i-l_1^{i,1})\alpha$ to the cut. That is, the right hand side of 
\eqref{eq:SecureCutBound} evaluates to $0$ for the $l_1^{i,1}$ number of eavesdropped nodes, and evaluates to $\alpha$ for the remaining $m_i-l_1^{i,1}$ number of nodes of first type in $i$th repair group.
The second type, on the other hand, consists of $n_i-m_i$ nodes in repair group $i$, out of
which $l_1^{i,2}$ of them are eavesdropped. As the amount of the data
downloaded during a node repair is equal to per node storage at the MBCR point,
the nodes that are eavesdropped do not add a value
to the cut. (For node $j$ in $i$th repair group, if  $j \in \Kc$ and $j$ is an eavesdropped node of second type, then right hand side of \eqref{eq:SecureCutBound}  evaluates as $H(\sv_j|\sv_1,\cdots,\sv_{j-1},\sv_{\Ec_1}) = H(\dv_j|\sv_1,\cdots,\sv_{j-1},\sv_{\Ec_1})=0$. This follows as one can get $\sv_j$ from $\dv_j$ and vice versa for the MBCR point, and $\sv_j \subseteq \sv_{\Ec_1}$. Therefore, the contribution of these nodes to the cut through their download links evaluates to zero in value.)
Consider the remaining $n_i-m_i-l_1^{i,2}$ number of
nodes in the repair group $i$, denoted as $\Nc_i$. These nodes contact $d$ live nodes, $\sum\limits_{r=0}^{i-1}n_r$
number of these contacted nodes belong to the previously repaired groups.
In addition, these nodes contact $t-n_i$ nodes that are previously repaired but not contacted by DC. 
For these nodes, the right hand side of \eqref{eq:SecureCutBound} can be evaluated as 

\begin{align*}
\sum\limits_{j\in\Nc_i}H(\sv_j|\sv_1,\cdots,\sv_{j-1},\sv_{\Ec_1}) &\overset{(a)}{=} \sum\limits_{j\in\Nc_i} H(\dv_j|\sv_1,\cdots,\sv_{j-1},\sv_{\Ec_1}) \\
&\overset{(b)}{=} H(\dv_{\Nc_i}|\sv_{\Pc},\sv_{\Ec_1}) \\
&\overset{(c)}{\leq} |\Nc_i|\left((d-\sum\limits_{r=0}^{i-1}n_r)\beta+(t - n_i)\beta'\right), \\
\end{align*}
where (a) follows as one can get $\sv_j$ from $\dv_j$ and vice versa for the MBCR point, (b) follows by summing the conditional entropy terms over $j\in\Nc_i$ where $\sv_{\Pc}$ denotes the set of nodes that are repaired before repairing the ones in $\Nc_i$, and (c) follows by upper bounding $H(\dv_{\Nc_i}|\sv_{\Pc},\sv_{\Ec_1})$ by assuming each node in $\Nc_i$ receives independent symbols from the previously repaired nodes and from the first $t-n_i$ nodes of this repair group. With $|\Nc_i|=n_i-m_i-l_1^{i,2}$, we have

\begin{figure}[t]
\centering
\includegraphics[width=0.4\columnwidth]{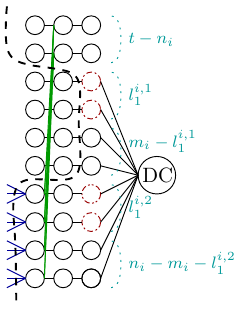}
\caption{The repair scenario considered to obtain an upper bound on secure file size for MBCR codes. Only a single repair group with $t$ number of nodes is shown. First $t-n_i$ nodes are not contacted by $\rm{DC}$. The following $l_1^{i,1}$ nodes are eavesdropped, $m_i-l_1^{i,1}$ nodes are contacted by $\rm{DC}$ with a cut of type $(x^{\rm in},x^{\rm co}|x^{\rm out})$, $l_1^{i,2}$ nodes are eavesdropped, and $n_i-m_i-l_1^{i,2}$ nodes are contacted by $\rm{DC}$ with a cut of type $(|x^{\rm in},x^{\rm co},x^{\rm out})$. For the last $n_i-m_i$ nodes, only the symbols downloaded from the first $t-n_i$ nodes contribute to the flow.}
\label{fig:MBCRCut}
\end{figure}

\begin{eqnarray}\label{eq:Cutsetsum}
\mathcal{M}^{s} \leq \sum\limits_{i=0}^{\mu-1} \left(\left(m_i-l_1^{i,1}\right)\alpha
+\left(n_i-m_i-l_1^{i,2}\right) C_i\right),
\end{eqnarray}
where

\begin{eqnarray}
C_i=\left(d-\sum\limits_{r=0}^{i-1}n_r\right)\beta
+
\left(t-n_i\right)\beta'.
\end{eqnarray}
We remark that the cut-value, i.e., the right hand side of \eqref{eq:Cutsetsum}, is minimized when we have
$$\sum\limits_{i=0}^{\mu-1} (l_1^{i,1} + l_1^{i,2}) = \ell_1.$$

We consider two scenarios in \eqref{eq:Cutsetsum}, (i) $m_i=0, l_1^{i,1}=0, l_1^{i,2}=l_1^{i}$ (for which the right hand side of \eqref{eq:Cutsetsum} evaluates to $(n_i-l_1^{i})C_i$) and
(ii) $m_i=n_i, l_1^{i,1}=l_1^{i}, l_1^{i,2}=0$ (for which the right hand side of \eqref{eq:Cutsetsum} evaluates to $(n_i-l_1^{i})\alpha$).  Hence, we obtain
\begin{align}\label{eq:Cutsetsum2}
\mathcal{M}^{s} \leq \sum\limits_{i=0}^{\mu-1}
\Bigg((n_i-l_1^i) \times
\min \Bigg\{\alpha,
\left(d-\sum\limits_{r=0}^{i-1}n_r\right)\beta
+
\left(t-n_i\right)\beta'
\Bigg\}\Bigg),
\end{align}
where $\sum\limits_{i=0}^{\mu-1}l_1^i=\ell_1$.

Note that, at the MBCR point, we have
\begin{eqnarray}\label{eq:CoopMBCR}
\alpha=d \beta + (t-1) \beta'.
\end{eqnarray}

Utilizing this, we consider the following case for
\eqref{eq:Cutsetsum2}.

\textbf{Case 1:} $\mu=k$, $n_i=1$, $\forall i=0, \cdots, k-1$. This case corresponds to a repair scenario where the data collector contacts only one node from each repair group. Accordingly, we have the following bound.

\begin{eqnarray*}
\mathcal{M}^{s} &\leq& \sum\limits_{i=0}^{k-1}
(1-l_1^i) \left((d-i)\beta + (t-1)\beta'\right)
\end{eqnarray*}
Here, the minimum cut value corresponds to having
$l_1^i=1$ for  $i=0,1,\cdots, \ell_1-1$, and
$l_1^i=0$ for $i = \ell_1,\ldots, k-1$. Hence, we get
\begin{eqnarray*}
\mathcal{M}^{s} &\leq& \sum\limits_{i=\ell_1}^{k-1}
(d-i)\beta + (t-1)\beta',
\end{eqnarray*}
from which we obtain
\begin{eqnarray}\label{eq:CoBound1}
\mathcal{M}^{s} \leq
\frac{(k-\ell_1)(2d-k-\ell_1+1)}{2}\beta + (k-\ell_1)(t-1)\beta'.
\end{eqnarray}

The bound in \eqref{eq:CoBound1} evaluates to the stated bound at the MBCR point.
\end{proof}

\begin{remark}
One can also consider the following repair scenarios in \eqref{eq:Cutsetsum2}.

\textbf{Case 2:} If $t\geq k$, $\mu=1$, $n_0=k$. Here, data collector contacts to nodes belonging to a single repair group. Accordingly, we have

\begin{eqnarray}
\mathcal{M}^{s} &\leq& (k-\ell_1) \left(d\beta + (t-k)\beta'\right). \label{eq:CoBound2}
\end{eqnarray}

\textbf{Case 3:} If $t<k$, $\mu=\floorb{k/t}+1$, $n_i=t$ for
$i=0,\cdots,\mu-2$, and $n_{\mu-1}=k-\floorb{k/t}t$. Here, data collector contacts to every node in every repair group (except for the last repair group). Denoting $a\triangleq \floorb{k/t}$ and $b\triangleq k-at$, so that
$k=at+b$, from \eqref{eq:Cutsetsum2}, we obtain
\begin{align}
\mathcal{M}^{s} \leq  \min\limits_{l_1^i\leq t \textrm{ s.t. } 
\sum\limits_{i=0}^{a} l_1^i=\ell_1} \: 
\sum\limits_{i=0}^{a-1}
(t-l_1^i) (d-it)\beta  
+ (b-l_1^a)\left\{(d-at)\beta + (t-b)\beta'\right\}.\label{eq:CoBound3}
\end{align}

We observe that both \eqref{eq:CoBound2} and \eqref{eq:CoBound3}
evaluate to loose bounds compared to that of \eqref{eq:CoBound1}. (For some parameters, Case 2/3 provides the same bound as that in Case 1. But, in general, Case 2/3 gives loose bound compared to Case 1. Details of this analysis are omitted for brevity.)
\end{remark}

In the following sections, we show that the bound given by \eqref{eq:CoBound1} is achievable, and hence it is the best bound that one can obtain for the MBCR point. That is, the repair scenario considered for Case 1 above is one of the worst cases and results in the tightest bound for the MBCR point under the secrecy constraints for all parameters. (For some parameters, Case 2/3 above represents an equivalently challenging bottleneck repair scenario for security in MBCR codes as well.)


\begin{figure*}
\centering
\includegraphics[width=0.8\textwidth]{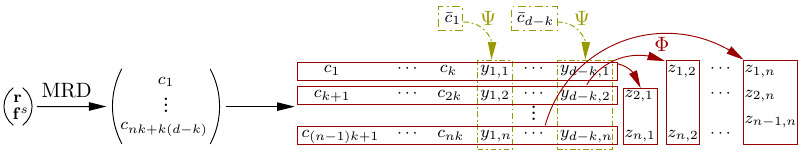}
\caption{The data $(\rv,\fv^s)$ is encoded with MRD coding into codeword symbols $c_1, \cdots, c_{nk+k(d-k)}$. Each row on the right represents a node in DSS, and the first $nk$ symbols are placed to $n$ nodes uniformly without additional coding. The remaining symbols are represented as $\bar{\cv}_j= (c_{nk+(j-1)k+1},\ldots, c_{nk+jk})$ for $j=1,\cdots,d-k$. Then, $\bar{\cv}_j$, of length-$k$, is encoded into $\yv_j = (y_{j, 1},\ldots, y_{j, n})$ using an MDS code, generator matrix of which is denoted by $\Psi$. Finally, $d$ primary symbols of node $i$, given by $\{c_{(i-1)k+1},\ldots, c_{ik}, y_{1,i},\ldots y_{d-k,i}\}$, is encoded into ($n-1$) symbols long codeword $\zv_{i} = (z_{1,i},\ldots, z_{i-1, i}, z_{i+1, i},\ldots, z_{n,i})$ using $\Phi$ as a generator matrix; and $z_{j,i}$ is placed at node $j\neq i$.}
\label{fig:SecureMBCR}
\end{figure*}

\subsection{Code construction for secure MBCR when $n=d+t$}
\label{sec:MBCR}

In this section, we focus on the special case $n=d+t$. (The case of $n>d+t$ will
be considered in the following section.)
We consider secrecy pre-coding of the data at hand before storing 
it to DSS using an MBCR code. We establish this pre-coding with maximum rank
distance (MRD) codes introduced in Section~\ref{sec:MRD}. 
(Please refer to Fig.~\ref{fig:SecureMBCR}.)

Consider the \emph{normalized} MBCR point given by 
$\mathcal{M}=k(2d-k+t)$, $\beta'=1$, $\beta=2$, $\alpha=\gamma=2d+t-1$,
$\mathcal{M}^{s}=k(2d-k+t)-\ell_1(2d-\ell_1+t)$, and $n=d+t$.
(The case of $\beta'>1$ can be obtained by implementing independent codes parallelly in the system.)
We use MRD codes with $N=K=\Mc$; hence, the rank distance
bound $D\leq N-K+1$ is saturated at $D=1$. Accordingly, we utilize 
$[\Mc,\Mc,1]$ MRD codes over $\F_{q^m}$, which maps length $\Mc$
vectors (each element of it being in $\F_{q^m}$) to length $\Mc$
codewords in $\F_{q^m}^{\Mc}$ (with $m\geq N=\Mc$).
The coefficients of the underlying linearized polynomial ($f(g)$) 
are chosen by $\mathcal{M}-\mathcal{M}^{s}$
random symbols denoted by $\rv\in \F_{q^m}^{\cM-\cM^s}$ and 
$\mathcal{M}^{s}$ secure data symbols denoted by $\fv^s\in \F_{q^m}^{\cM^s}$.
(That is, $\uv=(\rv,\fv^s)$ in \eqref{eq:MRDPoly}.)
The corresponding linearized polynomial $f(g)$ is evaluated at $\mathcal{M}$ points
$\{g_1$,\ldots, $g_{\mathcal{M}}\}$, which are linearly independent 
over $\mathbb{F}_q$. We denote these as $c_j=f(g_j)$ for 
$j=1,\cdots,\mathcal{M}$. This finalizes the secrecy pre-coding step.

The second encoding step is based on the encoding scheme for cooperative 
repair proposed in~\cite{Jiekak:CROSS12}.
(Here, we will summarize file recovery and node repair processes
for the case of MRD pre-coding, and provide the proof of security.)
Split the $\mathcal{M}$ symbols into two parts 
a) $c_1$ to $c_{nk}$,
and b) $c_{nk+1}$ to $c_{nk+k(d-k)}$. (Note that $n=d+t$ and $\mathcal{M}=nk+k(d-k)$.) 
The first part is divided into $n$ groups of $k$ symbols, and
stored in $n$ nodes. Here, node $i$ stores
$c_{(i-1)k+1}$ to $c_{ik}$. The second part is divided
into $d-k$ groups of $k$ symbols. These symbols are
encoded with an $(n,k)$ MDS code, and stored on $n$ nodes.
In particular, $\yv_j = (y_{j,1},\ldots, y_{j,n})$ is generated from symbols
$\bar{\cv}_j = (c_{nk+(j-1)k+1},\ldots, c_{nk+jk})$ by utilizing some MDS encoding matrix $\Psi$ of size $k \times n$; then, $y_{j,i}$ is stored
at node $i$, for $j=1,\cdots, d-k$.
Node $i$, having stored $\{c_{(i-1)k+1},\ldots, c_{ik}, y_{1,i},\ldots y_{d-k,i}\}$, 
which is referred to as the primary data of node $i$,
encodes these symbols
using an $(n-1,d)$ MDS code that has a generator matrix given by a (generalized) Cauchy matrix 
$\Phi$ of size $d\times (n-1)$. 
(This choice of $\Phi$ ensures that $[\Id_d~\Phi]$
is a generator matrix for an $(n+d-1,d)$ MDS code~\cite{Roth:generator85}.) Let $\zv_i = (z_{1,i},\ldots, z_{i-1, i}, z_{i+1, i},\ldots, z_{n,i})$ denote the $(n-1)$ symbols long codeword obtained by encoding the primary data of node $i$. These $n-1$ symbols are stored in every other node one-by-one. In particular, node $j \neq i$ stores $z_{j, i}$.
We call $\{z_{j,i}: i\in[1:n], i\neq j\}$ as the secondary data. This procedure is repeated
for every node, so that each node $i$ stores
$\{c_{(i-1)k+1}, \ldots, c_{ik},y_{1,i},\ldots,y_{d-k,i},
z_{i,1},\ldots, z_{i,i-1}, z_{i,i+1},$ $\ldots, z_{i,n}\}$,
and hence total number of symbols stored at
each node is $k+(d-k)+(n-1) = d+n-1 = 2d+t-1= \alpha$.

\emph{File recovery at DC:} DC connects to any $k$
nodes, without loss of generality we assume the first
$k$ nodes. From $y_{j,1:k}$, DC can obtain
$c_{nk+(j-1)k+1}, \cdots, c_{nk+jk}$, for each $j=[1:d-k]$. 
It can re-encode these into $y_{j,1:n}$ using the MDS code, and
obtain the other $y$ symbols at the remaining nodes.
Then, for each $i\in [k+1:n]$, DC can use the MDS property
of $[\Id_d \: \Phi]$, to obtain $c_{(i-1)k+1}, \cdots, c_{ik}$
symbols of node $i$ from the $k$ secondary data symbols
of the contacted nodes, i.e., $z_{j,i}$ for $j=[1:k]$,
 and additional $d-k$ symbols, $y_{j,i}$ for $j=[1:d-k]$. Having obtained
$c_1, \cdots, c_{\cM}$, DC can perform interpolation to
solve for both data and random coefficients.

\emph{Node repair:} Assume that the first $t$ nodes fail.
From the secondary data stored in the remaining $d=n-t$
nodes, $z_{t+1,i},\cdots,z_{n,i}$, one can recover
$c_{(i-1)k+1}, \cdots, c_{ik}$ and $y_{1,i},\cdots,y_{d-k,i}$
for node $i=1,\cdots, t$. (This corresponds to sending
$1$ symbol from each of $d$ nodes to each of the $t$ nodes.)
Then, to recover the secondary data stored at each node under repair,
say for the node $j=1,\cdots,t$, every other node, i.e., nodes $i\neq j$,
including the nodes under repair, computes and sends its corresponding 
encoded primary data, i.e., $z_{j,i}$, to node $j$. 
(This corresponds to sending $1$ symbol
from each node to each of the $t$ nodes.)
This achieves $\beta=2$ and $\beta'=1$ symbols for the
repair procedure.

\begin{figure*}[t]
\centering
\includegraphics[width=0.8\textwidth]{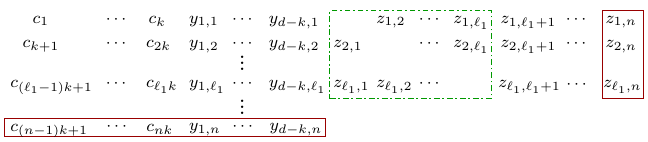}
\caption{The eavesdropped nodes are represented by the first $\ell_1$ rows, each row corresponding to a node. The symbols in the dashed-dotted (green) box are linear combinations of the symbols in $\Cc=\{c_{1}, \ldots, c_{\ell_1k}\}$ and $\Yc=\{y_{1,1},\ldots, y_{d-k,1}, \cdots, y_{1,\ell_1},\ldots, y_{d-k,\ell_1}\}$. The remaining symbols denoted with $z_{1:\ell_1,\ell_1+1:n}$ are linearly independent of $\Xc$ and $\Yc$ due to encoding with the matrix $\Phi$, where $z_{1:\ell_1,i}$ is generated from $\{c_{(i-1)k+1},\ldots, c_{ik}, y_{1,i},\ldots y_{d-k,i}\}$ for $i=\ell_1,\cdots,n$. For example, the last row, primary symbols of node $n$ generates the last column, i.e., the eavesdropped symbols $z_{1:\ell_1,n}$. (See Fig.~\ref{fig:SecureMBCR} for details of the encoding steps.)}
\label{fig:EveSymbols}
\end{figure*}

\emph{Security:} Consider that the eavesdropper is observing the
first $\ell_1$ nodes. (See Fig.~\ref{fig:EveSymbols}.)
Let $\Cc\triangleq\{c_{1}, \ldots, c_{\ell_1k}\}$,
$\Yc\triangleq\{y_{1,1},\ldots, y_{d-k,1}, \cdots, y_{1,\ell_1},\ldots, y_{d-k,\ell_1}\}$,
$\Zc\triangleq\{z_{j,i} \textrm{ for } j= 1,\ldots,\ell_1,
\textrm{ and } i=\ell_1+1,\cdots,n$\}.

Due to the coding scheme, the symbols $\{z_{j,i}\}$ for $j=1,\cdots,\ell_1$; $i=1,\cdots,\ell_1$; and $j\neq i$ are linear combinations of the symbols in $\Cc\cup\Yc$. 
In addition, the symbols in the sets $\Cc,\Yc,\Zc$ correspond to linearly independent evaluation points. This follows due to the code construction as detailed in the following. We remark that symbols in $\Cc$ are clearly independent of the ones in $\Yc \cup \Zc$. However, it is not clear at first sight whether the symbols in $\Yc$ and $\Zc$ are independent. Note that, for example, the symbols $y_{1,1:n}$ are dependent, i.e., each one can be uniquely determined by any other $k$ of them due to MDS coding by $\Psi$. And, symbols $y_{1:d-k,\ell+1:n}$ generates $\Zc$. But, the generation of $\Zc$ is \emph{together with} the symbols in $c_{\ell_1k+1:nk}$. (See Fig.~\ref{fig:EveSymbols}.) 

In particular, we have the following,

\begin{align*}
&[\tilde{C} \: \tilde{Y}]\tilde{\Phi}=Z,
\end{align*}
where
\begin{align*}
\tilde{C}=
\left[
\begin{array}{ccc}
c_{\ell_1k+1} & \cdots & c_{(\ell_1+1)k} \\
\vdots & \ddots & \vdots \\
c_{(n-1)k+1} & \cdots & c_{nk} 
\end{array}
\right],
\end{align*}

\begin{align*}
&\tilde{Y}=
\left[ \begin{array}{ccc}
y_{1,\ell_1+1} & \cdots & y_{d-k,\ell_1+1} \\
\vdots & \ddots & \vdots \\
y_{1,n} & \cdots & y_{d-k,n} \\
\end{array} \right],\\
&Z=
\left[ \begin{array}{ccc}
z_{1,\ell_1+1} & \cdots & z_{\ell_1,\ell_1+1} \\
\vdots & \ddots & \vdots \\
z_{1,n} & \cdots & z_{\ell_1,n} \\
\end{array} \right],
\end{align*}
and $\tilde{\Phi}$ is the corresponding $\ell_1$ columns of $\Phi$. (In the example above, we consider the first $\ell_1$ columns.) Consider 
$$\tilde{\Phi}=\left[\begin{array}{c}
\tilde{\Phi}^u \\ 
\tilde{\Phi}^l
\end{array}\right]
$$
with $k\times \ell_1$ matrix $\tilde{\Phi}^u$, and, accordingly, $\tilde{C}\tilde{\Phi}^u + \tilde{Y}\tilde{\Phi}^l=Z$. Here, the symbols in $\tilde{C}\tilde{\Phi}^u$ are linearly independent. (This follows, as appending upper $k$ rows of any $k-\ell_1$ number of remaining columns of $\Phi$ to $\tilde{\Phi}^u$ will constitute a $k\times k$ submatrix of $\Phi$, which is non-singular. Denoting this matrix as $\Phi'=[\tilde{\Phi}^u \cdots]$, where $\cdots$ representing the added elements of $\Phi$, we observe that all symbols in $\tilde{C}\Phi'$ are linearly independent, which implies the linear independence of symbols in $\tilde{C}\tilde{\Phi}^u$.) Then, as the symbols in the matrix $\tilde{C}$ are independent of the ones in the set $\{y_{1:d-k,1:n}\}$, it follows that symbols in the matrix $Z=\tilde{C}\tilde{\Phi}^u + \tilde{Y}\tilde{\Phi}^l$, i.e., the symbols in the set $\Zc$, and the symbols in the set $\Yc$ are linearly independent.

Therefore, due to the linearized property of the code, the eavesdropper, observing
$\ell_1\alpha=\ell_1 (2d+t-1)$ symbols, has evaluations of the polynomial $f(\cdot)$ at $\ell_1(2d+t-\ell_1)$ linearly independent points. Using the secure data symbols, together with
interpolation from these $\ell_1(2d+t-\ell_1)$ symbols, the eavesdropper can
solve for $\ell_1(2d+t-\ell_1)$ random symbols. Then, denoting the
eavesdroppers' observation as $\ev$, we have $H(\rv|\ev,\fv^s)=0$.
Since $H(\ev)=H(\rv)$, from Lemma~\ref{thm:SecrecyLemma},
we obtain $I(\fv^s;\ev)=0$.

Now, using the upper bound given in Proposition~\ref{thm:MBCRbound},
we obtain the following result.
\begin{proposition}
The secrecy capacity at the MBCR point for a file size 
of $\mathcal{M}=k(2d-k+t)\beta'$ is given by
$\mathcal{M}^{s}=k(2d-k+t)\beta'-\ell_1(2d-\ell_1+t)\beta'$, if $n=d+t$.
\end{proposition}


\subsection{Does cooperation enhance/degrade security at MBCR?}

The repair bandwidth for cooperative regenerating codes  is defined by $\gamma=d\beta + (t-1)\beta'$. In this section, we analyze $\bar{\gamma}=\frac{\gamma}{\mathcal{M}^{s}}$, the ratio of repair bandwidth to
the secure file size, referred to as
the normalized repair bandwidth (NRBW).

Without the security constraints, for which $\ell_1=0$ in
Proposition~\ref{thm:MBCRbound}, we observe that, at the MBCR point, NRBW
is given by
\begin{equation*}
\bar{\gamma}(\ell_1=0) = \frac{2d+t-1}{k(2d-k+t)},
\end{equation*}
which is equal to
\begin{equation*}
\bar{\gamma}(\ell_1=0,n=d+t) = \frac{2n-t-1}{k(2n-k-t)}
\end{equation*}
for a system with $n=d+t$. Here, the classical
(i.e., non-cooperative) scenario corresponds to $t=1$
case, which has an NRBW of
\begin{equation*}
\bar{\gamma}(\ell_1=0,n=d+t, t=1) = \frac{2n-2}{k(2n-k-1)}.
\end{equation*}
Comparing the last two equations, we see that
$$\bar{\gamma}(\ell_1=0,n=d+t)\geq \bar{\gamma}(\ell_1=0,n=d+t, t=1),$$
with equality iff $t=1$ (for $k>1$). Therefore, without the security
constraints, having simultaneous repairs of size greater
than $1$ actually increases the normalized repair bandwidth.
This nature of cooperation also results in the conclusion that
deliberately delaying the repairs does not bring additional
savings~\cite{Kermarrec:Repairing11}. (This observation is proposed
for both MBCR and MSCR points in~\cite{Kermarrec:Repairing11}
with an analysis of derivative of $\gamma$ with respect to
$t$. Here, we provide an analysis with NRBW.)

We revisit the above conclusion under security constraints.
The question is whether the cooperation (i.e., having a system
with multiple failures, or deliberately delaying the repairs)
results in a loss/gain in secure DSS. We have
\begin{equation*}
\bar{\gamma}(n=d+t) = \frac{2n-t-1}{k(2n-k-t)-\ell_1(2n-\ell_1-t)},
\end{equation*}
and 
\begin{equation*}
\bar{\gamma}(n=d+t,t=1) = \frac{2n-2}{k(2n-k-1)-\ell_1(2n-\ell_1-1)}.
\end{equation*}
A calculation similar to above shows that 
$\bar{\gamma}(n=d+t)\geq \bar{\gamma}(n=d+t,t=1)$ with equality if and only if $t=1$ (for $k>\ell_1\geq0$). This shows that NRBW for the case $t>1$ is strictly greater than that of $t=1$ when $n=d+t$ for $\ell_1<k$. The MBCR points given in
Proposition~\ref{thm:MBCRbound} for codes satisfying
$0\leq \ell_1 < k < n$, $d\geq k$, and $d=n-t$
are given in Table~\ref{tab:Coop} in Appendix~\ref{app:CoopTable}.
As evident from the table, we see that cooperation does not
bring additional savings for secure DSS at the MBCR point when
$d+t=n$. This in turn means that one may not delay
the repairs to achieve a better performance than that of single
failure-repair if $d$ is chosen such that $n=d+t$ for a given $t,n$.
However, if the downloads within the cooperative group are less
costly compared to the downloads from the live nodes, then
delaying repairs would be beneficial in reducing the total cost.
We will revisit this analysis for codes having $n>d+t$ in the next subsection.


\subsection{General code construction for secure MBCR}

The code construction presented in Section III-B has the requirement of $d=n-t$.
However, for practical systems, it may not be possible that
a failed node connects to all the remaining nodes. This brings
the necessity of code constructions for $d<n-t$. Remarkably,
for a fixed $(n,k,d,{\mathcal{M}})$, increasing $t$ can
reduce the repair bandwidth in the secrecy scenario we consider here. This is
reported in~\cite{Shum:Exact11} for DSS without secrecy constraints.
Hence, for a fixed $d$, delaying the repairs can be advantageous,
e.g., when there is a limit on the number of live nodes that can be contacted for a node repair.
In the following, we present a general construction which
works for any parameters, in particular for $n>d+t$.

The construction is based on the MBCR code proposed in \cite{Wang:Exact12}. In \cite{Wang:Exact12}, a bivariate polynomial is constructed using ${\mathcal{M}} = k(2d + t - k)$ message symbols (over $\F_q$) as the coefficients of the polynomial:
\begin{align}
\label{eq:mbcr_poly}
F(Y,Z) = \sum_{\substack{0 \leq i < k,\\ 0 \leq j < k}}a_{ij}Y^iZ^j + \sum_{\substack{0\leq i < k,\\ k \leq j < d+t}}b_{ij}Y^iZ^j  + \sum_{\substack{k\leq i < d,\\ 0 \leq j < k}}c_{ij}Y^iZ^j
\end{align}

\begin{figure*}
\footnotesize
\centering
\begin{tabular}{|c|c|c|c|c|c|c|c|c|}
  \hline
  $F(y_1,z_1)$ & $F(y_1,z_2)$ & $\cdots$ & $F(y_1,z_{\ell_1})$ & $\cdots$
& $F(y_1,z_{d+t})$ & & & \\\hline
  $F(y_2,z_1)$&$F(y_2,z_2)$&$\cdots$&$F(y_2,z_{\ell_1})$&$\cdots$
& $F(y_2,z_{d+t})$ & $\textcolor{blue}{F(y_2,z_{d+t+1})}$ & &\\\hline
  $\cdots$ & & & &
& & & & \\\hline
  $F(y_{\ell_1},z_1)$ & $F(y_{\ell_1},z_2)$ & $\cdots$ & $F(y_{\ell_1},z_{\ell_1})$ & $\cdots$
& $F(y_{\ell_1},z_{d+t})$ & $\textcolor{blue}{F(y_{\ell_1},z_{d+t+1})}$ & $\textcolor{blue}{\cdots}$ & $\textcolor{blue}{F(y_{\ell_1},z_{d+t-1+\ell_1})}$ \\\hline
  $\cdots$ & & & &
& & & & \\\hline
  $F(y_{d},z_1)$ & $F(y_{d},z_2)$ & $\cdots$ & $F(y_{d},z_{\ell_1})$ &
& & & & \\\hline
   & $\textcolor{blue}{F(y_{d+1},z_2)}$ & $\textcolor{blue}{\cdots}$ & $\textcolor{blue}{F(y_{d+1},z_{\ell_1})}$ &
& & & & \\\hline
  & & $\textcolor{blue}{\cdots}$  & $\textcolor{blue}{\cdot}$ &
& & & & \\\hline
   & & & $\textcolor{blue}{F(y_{d+\ell_1-1},z_{\ell_1})}$ &
& & & & \\\hline
\end{tabular}
\caption{Observed symbols at the eavesdroppers for a given $\ell_1$.}
\label{fig:MBCREave}
\end{figure*}
\normalsize

Here, $\{a_{ij}\}$, $\{b_{ij}\}$, and $\{c_{ij}\}$ denote $\Mc$ message symbols. Given $q > n$, two set of $n$ distinct points, $\{y_1, y_2,\ldots, y_n\}$ and $\{z_1, z_2,\ldots, z_n\}$, are chosen. The $i${th} node in the DSS store the following $2d + t - 1$ evaluations of polynomial $F(Y,Z)$.
\begin{eqnarray}
\label{eq:MBCR_node_content}
F(y_i,z_i), F(y_i,z_{i\oplus1}),\ldots, F(y_i,z_{i\oplus(d+t-1)}) \n  \\
F(y_{i\oplus1},z_i), F(y_{i\oplus2},z_{i}),\ldots, F(y_{i\oplus(d-1)},z_{i}),
\end{eqnarray}
where $\oplus$ denotes addition modulo $n$. The first $d+t$ evaluations at node $i$ can be seen as the evaluations of the univariate polynomial $f_i(Z) = F(y_i,Z)$ of degree at most $d+t-1$ at $d+t$ points. This uniquely defines the polynomial $f_i(Z)$. Similarly, the first evaluation in (\ref{eq:MBCR_node_content}), $F(y_i,z_i)$, along with last $d-1$ evaluations, $F(y_{i\oplus1},z_i),\ldots, F(y_{i\oplus(d-1)},z_{i})$, uniquely define the univariate polynomial $g_i(Y) = F(Y,z_i)$ of degree at most $d-1$. This property of the proposed bivariate polynomial based coding scheme is utilized for the exact node repair and data reconstruction processes at the MBCR point. (We refer to \cite{Wang:Exact12} for details.)

In order to get an $(\ell_1,0)$-secure code at the MBCR point, we rewrite the polynomial in (\ref{eq:mbcr_poly}) as follows:
\begin{align}
\label{eq:secureMBCR_ploy}
F(Y,Z) &=\sum_{\substack{0 \leq i < \ell_1, \\ 0 \leq j < \ell_1}}a_{ij}Y^iZ^j + \sum_{\substack{0 \leq i < \ell_1, \\ \ell_1 \leq j < k}}a_{ij}Y^iZ^j& \n \\
        &+ \sum_{\substack{\ell_1 \leq i < k,\\ 0 \leq j < \ell_1}}a_{ij}Y^iZ^j  + \sum_{\substack{\ell_1 \leq i < k, \\ \ell_1 \leq j < k}}a_{ij}Y^iZ^j& \n \\
        &+ \sum_{\substack{0\leq i < \ell_1,\\ k \leq j < d+t}}b_{ij}Y^iZ^j + \sum_{\substack{\ell_1 \leq i < k,\\ k \leq j < d+t}}b_{ij}Y^iZ^j& \n \\
        &+ \sum_{\substack{k\leq i < d,\\ 0 \leq j < \ell_1}}c_{ij}Y^iZ^j + \sum_{\substack{k\leq i < d,\\ \ell_1 \leq j < k}}c_{ij}Y^iZ^j&
\end{align}
Next, we choose $\ell_1^2 + \ell_1(k-\ell_1) + (k-\ell_1)\ell_1 + \ell_1(d+t - k) + (d-k)\ell_1 = \ell_1(2d + t - \ell_1)$ coefficients of $F(Y,Z)$, $a_{0:\ell_1-1, 0:\ell_1-1} \cup  a_{0:\ell_1-1, \ell_1:k-1} \cup a_{\ell_1:k-1, 0:\ell_1-1}\cup  b_{0:\ell_1-1, k:d+t-1}\cup c_{k:d-1, 0:\ell_1-1}$, to be symbols drawn uniformly at random from $\F_{q}$ in an i.i.d. manner. Remaining $k(2d + t -k) - \ell_1(2d + t -\ell_1) = \mathcal{M}^{s}$ coefficients of $F(Y,Z)$ are chosen to be the data symbols $\fv^{s}$ that need to be stored on the DSS. Each node $i \in [n]$ stores the evaluation of $F(Y,Z)$ as illustrated in (\ref{eq:MBCR_node_content}). It follows from the description of the coding scheme of \cite{Wang:Exact12} in the beginning of this subsection that the resulting coding scheme is an exact repairable code at the MBCR point.

Next, we show that the proposed scheme is indeed $(\ell_1,0)$-secure. Let $\ev$, $\fv^s$, and $\rv$ denote the data observed by an eavesdropper, the original data to be stored, and the randomness added to the original data before encoding, respectively. It is sufficient to show (i) $H(\ev) \leq H(\rv)$ and (ii) $H(\rv| \fv^s, \ev) = 0$ in order to establish the secrecy claim (see Lemma~\ref{thm:SecrecyLemma}). To argue the first requirement, noting that number of eavesdropped
symbols are $\ell_1\alpha=\ell_1(2d+t-1)$, we will show that $\ell_1^2-\ell_1$
number of these are linearly dependent on the remaining ones.
The eavesdropper, without loss of generality considering the
first $\ell_1$ nodes as eavesdropped nodes, observes the symbols given in Fig.~\ref{fig:MBCREave}.
Due to the code construction, each row above represents
evaluations of a polynomial of degree less than $d+t$ and each column
represents a polynomial of degree less than $d$. Hence,
we observe that each of the symbols denoted with a colored
font in the matrix of Fig.~\ref{fig:MBCREave} is a linear combination of the
remaining ones. Therefore, $H(\ev) \leq \ell_1\alpha-\ell_1(\ell_1-1)=H(\rv)$.

In order to show that second requirement also holds, we present a method to decode randomness $\rv$ given $\fv^s$ and data stored on any $\ell_1$ nodes. Once we know the data symbols $\fv^s$, we can remove the monomials associated to data symbols in $F(Y,Z)$ and the contribution of these monomials from the polynomial evaluations stored on DSS. Let $\widehat{F}(Y,Z)$ denote the bivariate polynomial that we obtain by removing the data monomials:
\begin{align}
\label{eq:MBCR_eaves_poly1}
\widehat{F}(Y,Z) &=\sum_{\substack{0 \leq i < \ell_1,\\ 0 \leq j < \ell_1}}a_{ij}Y^iZ^j + \sum_{\substack{0 \leq i < \ell_1,\\ \ell_1 \leq j < k}}a_{ij}Y^iZ^j& \n \\
        &+ \sum_{\substack{\ell_1 \leq i < k,\\ 0 \leq j < \ell_1}}a_{ij}Y^iZ^j  + \sum_{\substack{0\leq i < \ell_1,\\ k \leq j < d+t}}b_{ij}Y^iZ^j& \n \\
        &+ \sum_{\substack{k\leq i < d,\\ 0 \leq j < \ell_1}}c_{ij}Y^iZ^j&
\end{align}
$\widehat{F}(Y,Z)$ can be rewritten as:
\begin{align}
\label{eq:MBCR_eaves_poly2}
\widehat{F}(Y,Z) &=\sum_{\substack{0 \leq i < \ell_1,\\ 0 \leq j < \ell_1}}\hat{a}_{ij}Y^iZ^j + \sum_{\substack{0 \leq i < \ell_1,\\ \ell_1 \leq j < d+t}}\hat{b}_{ij}Y^iZ^j
              + \sum_{\substack{\ell_1 \leq i < d,\\ 0 \leq j < \ell_1}}\hat{c}_{ij}Y^iZ^j
\end{align}
where
$\hat{a}_{0:\ell_1-1, 0:\ell_1-1} = a_{0:\ell_1-1, 0:\ell_1-1},
 \hat{b}_{0:\ell_1-1, \ell_1:k-1} =  a_{0:\ell_1-1, \ell_1:k-1},
 \hat{b}_{0:\ell_1-1, k:d+t-1}=  b_{0:\ell_1-1, k:d+t-1},
 \hat{c}_{\ell_1:k-1, 0:\ell_1-1} = a_{\ell_1:k-1, 0:\ell_1-1},
\hat{c}_{k:d-1, 0:\ell_1-1}= c_{k:d-1, 0:\ell_1-1}.$

$\widehat{F}(Y,Z)$ in (\ref{eq:MBCR_eaves_poly2}) takes the same form as $F(Y,Z)$ in (\ref{eq:mbcr_poly}) with $k$ replaced with $\ell_1$. Now, the randomness $\rv$, coefficients of $\widehat{F}(Y,Z)$ in (\ref{eq:MBCR_eaves_poly2}), can be decoded from the data observed on $\ell_1$ nodes using the data reconstruction method described in \cite{Wang:Exact12}.
Thus, we obtain the following result.
\begin{proposition}
The secrecy capacity at the MBCR point for a file size 
of $\mathcal{M}=k(2d-k+t)$ is given by
$\mathcal{M}^{s}=k(2d-k+t)-\ell_1(2d-\ell_1+t)$ for any $n\geq d+t$.
\end{proposition}
We list some instances of this construction in Table~\ref{tab:Coop2}
in Appendix~\ref{app:CoopTable}. As evident from the table,
cooperation helps to reduce the repair bandwidth if $d<n-t$.
Thus, (secure) coding schemes for the case of $n>d+t$ are of significant
interest in order to reduce the repair bandwidth in cooperative
repair.


\section{Secure MSCR Codes}
\label{sec:mscr}

We first consider an upper bound on the secure file size for DSS employing minimum storage cooperative regenerating (MSCR) codes. We then utilize 
appropriate secrecy pre-coding mechanisms to construct achievable schemes for the upper bound.

\subsection{Upper bound on the secure file size}

At the MSCR point, nodes have minimum possible storage,
i.e., $\alpha=\frac{{\mathcal{M}}}{k}$.
Using the cut-set analysis given in Section~\ref{sec:sys_model},
one can obtain that the minimum repair bandwidth can
be attained with $\beta=\beta'=\frac{\alpha}{d-k+t}=\frac{{\mathcal{M}}}{k(d-k+t)}$.
(See also~\cite{Kermarrec:Repairing11,Oggier:Coding12}.) Therefore, the amount of data downloaded during a node repair is larger than per node storage $\alpha$ at the MSCR point. Keeping this in mind, we consider two kinds of eavesdropped nodes for security constraints: storage-eavesdropped nodes ($\Ec_1$) and download-eavesdropped nodes ($\Ec_2$).
Using the size of these sets we denote the eavesdropper
setting with $(\ell_1,\ell_2)$ as introduced in Section~\ref{sec:sys_model}.
Here, for a node in $\Ec_2$, an eavesdropper observes both the data downloaded from live nodes and from cooperating nodes (other nodes in the same repair group).

Similar to the analysis given in Section~\ref{subsec:mbcr_bound}, and utilizing Lemma~\ref{thm:SecureCutLemma}, we 
obtain the following bound on the secure file size at the MSCR point.

\begin{align}
\mathcal{M}^{s} \leq \sum\limits_{i=0}^{k-1}
(1-l_1^i-l_2^i) \times 
\min \Big\{ \alpha - I(\sv_{i};\dv_{i,\Ec_2}),
(d-i)\beta + (t-1)\beta'\Big\},
\label{eq:CoBound4}
\end{align}
where $\sv_i$ and $\dv_{i, \Ec_2}$ denote the data stored on node $i$ and data downloaded from node $i$ for repair of nodes in set $\Ec_2$, respectively. Here, for $i$th repair group, we consider $n_i=1$ number of nodes to be contacted by $\rm{DC}$. We assume that we have $l_1^i$ number of storage-eavesdropped nodes (from $\Ec_1$) and $l_2^i$ number
of download-eavesdroppers (from $\Ec_2$) in repair group $i$. Compared to the MBCR bounds, due to
eavesdroppers in $\Ec_2$, nodes that are not eavesdropped may leak information during their participation in the repair of a node having an download-observing eavesdropper. Thus, the values of the cuts of type $1$
include additional penalty terms $I(\sv_{i};\dv_{i,\Ec_2})$, counting
the leakage from the storage at the $i$th node to nodes indexed
with $\Ec_2$.
(That is, we consider a loose bound compared to that of \eqref{eq:SecureCutBound}, i.e.,
$H(\sv_i|\sv_1,\cdots,\sv_{i-1},\sv_{\Ec_1},\dv_{\Ec_2}) = H(\sv_i)-I(\sv_i;\sv_1,\cdots,\sv_{i-1},\sv_{\Ec_1},\dv_{\Ec_2})\leq
\alpha - I(\sv_i;\dv_{i,\Ec_2})$, as $\dv_{i,\Ec_2}\subset\dv_{\Ec_2}$.)
Considering the MSCR point values of $\alpha$, $\beta$, and
$\beta'$ given above, the second term inside each $\min\{ \}$ in \eqref{eq:CoBound4}
is larger than the first term.
Hence, considering that the first $k-\ell_1-\ell_2$ repairs are
eavesdropper-free, \eqref{eq:CoBound4} evaluates to the
following bound.

\begin{proposition}\label{thm:MSCRbound}
Cooperative regenerating codes operating at the MSCR
point with a secure file size of $\mathcal{M}^{s}$ satisfy
\begin{eqnarray*}
\mathcal{M}^{s} \leq \sum\limits_{i=0}^{k-\ell_1-\ell_2-1}
\alpha - I(\sv_i;\dv_{i,\Ec_2}),
\end{eqnarray*}
where the MSCR point is given by $\beta=\beta'$,
$\alpha=(d-k+t)\beta$, for a file size of ${\mathcal{M}}=k(d-k+t)\beta$.
In addition, at the MSCR point, one can bound $I(\sv_i;\dv_{i,\Ec_2})\geq \beta'=\beta$
and obtain the bound
$$\mathcal{M}^{s} \leq (k-\ell_1-\ell_2)(\alpha-\beta).$$
\end{proposition}


\subsection{Code construction for secure MSCR when $k=t=2$}

We consider an interference alignment approach based on the one proposed in~\cite{LeScouarnec:Exact12} with $k=t=2$.
For any $(n,k,d,t)$ with $d\geq k$ and $n=d+t$, we have $\alpha=d-k+t=n-2$ and ${\mathcal{M}}=k(d-k+t)=2(d-k+t)=2\alpha$ at the normalized MSCR point, i.e., $\beta = \beta'=1$. (The general case of $\beta>1$ can be obtained by a parallel utilization of independent codes with $\beta=1$.)  For $k = 2$, the bound given in Proposition~\ref{thm:MSCRbound} implies that the achievability of positive secure file size in the presence of an eavesdropper is possible
only when $(\ell_1,\ell_2)=(1,0)$ or $(\ell_1,\ell_2)=(0,1)$. Corresponding bounds are given by $\mathcal{M}^{s} \leq \alpha$
and $\mathcal{M}^{s} \leq \alpha-1$, respectively. (For the bound corresdpoing to $\ell_2 = |\Ec_2|=1$, $\dv_{i,\Ec_2}$ necessarily consists of one
symbol as $\beta=\beta'=1$, and the non-eavesdropped
node participates in the repair of the eavesdropped node
by sending $\beta$ or $\beta'$ symbols. Note that ${\rm DC}$ contacts only $k =  2$ nodes, and one of them is eavesdropped for the purpose of obtaining the upper bound here.)
In the following, we construct codes achieving the stated
bounds for both cases, hence establishing the secrecy capacity
when $k=t=2$. We show this with codes having $n=d+t$, i.e.,
all the surviving nodes participate in a node repair.
The construction can be extended to  cases with $n>d+t$
by following a similar approach and choosing a larger field size.

\textbf{Case $1$: $\mathcal{M}^{s} = \alpha$ when $(\ell_1,\ell_2)=(1,0)$}

Consider a large enough finite field $\F_q$ (conditions on the required field will be given later) which has $\omega$ as its generator,
$\alpha$ number of random symbols $r_1, \cdots, r_\alpha$, and
$\alpha$ number of secure information symbols $s_1, \cdots,
s_\alpha$. Both information and random symbols are uniformly
distributed over the field. We construct a file $\fv = (a_1,\ldots, a_{\alpha}, b_1,\ldots, b_{\alpha}) =(r_1, \ldots, r_\alpha, r_1+s_1, \ldots, r_\alpha+s_\alpha)$ of size $\Mc$ over $\F_q$. Our coding scheme is described by the following placement.
\begin{itemize}
\item Store $\av=(a_1,\cdots,a_\alpha)$ at the first node,
\item Store $\bv=(b_1,\cdots,b_\alpha)$ at the second node, and
\item Store $\rv_i=(a_1+\omega^{(i-1)~{\rm mod}~\alpha}b_1, \cdots, a_\alpha+\omega^{(i+\alpha-2)~{\rm mod}~\alpha}b_\alpha)$
at $i$th parity node, $i\in\{1,\cdots, \alpha = n - 2\}$.
\end{itemize}
Note that we implement one-time pad scheme of Vernam~\cite{Vernam:Cipher26} for the symbols represented as $b_i$ in this construction. We represent the stored symbols of $i$th parity node as
$\rv_i^T=\av^T + \Bm_i \bv^T$, where $\Bm_i$ is a diagonal matrix with its diagonal elements given by
$\{\omega^{(i-1)~{\rm mod}~\alpha},\ldots , \omega^{(i+\alpha-2)~{\rm mod}~\alpha}\}$.
Data collector $\rm{DC}$ can reconstruct the file $\fv$ by contacting to
any of the $k = 2$ nodes, and solving $\alpha$ groups of $2$ equations
over $2$ unknowns. (Each group gives $2$ equations over $2$ unknowns. For example, the first symbols for the systematic nodes are $a_1=r_1$ and $b_1=r_1+s_1$, and these two equations can be solved to obtain $(r_1,s_1)$.) From file $\fv$, it can then obtain
the secure symbols $s_1, \cdots, s_\alpha$. This establishes data reconstruction property of the code. Cooperative repair process is similar to that of the code proposed in~\cite{LeScouarnec:Exact12}, as summarized in the following.

Without loss of generality, we consider cooperative repair of two systematic nodes. Repairs of stages involving parity nodes can be performed as that of the systematic nodes after change of variables~\cite{LeScouarnec:Exact12}. The first systematic node downloads $\vv_{1,i}\rv_i^T=\vv_{1,i}\av^T + \zv \bv^T$ from $i$th parity node which stores $\rv_i^T=\av^T + \Bm_i \bv^T$. Here, $\vv_{1,i}=\zv \Bm_i^{-1}$ and $\zv\triangleq(1,\cdots,1)$. After this step, the first systematic node has the symbols $$\cv_1=\{\zv \Bm_1^{-1}\av^T + \zv\bv^T, \cdots, \zv \Bm_d^{-1}\av^T+\zv\bv^T\}.$$  Note that the repair process is such that the interference is aligned by having terms $\zv\bv^T$ in each of the symbols downloaded by the first systematic node. The second systematic node obtains $d = n -2$ symbols $\cv_2=\{\vv_{2,1}\rv_1^T,\cdots, \vv_{2,d}\rv_d^T\}$ from $d$ parity nodes. Here, the repair process is such that the interference is aligned by having terms $\zv\av^T$ in each of the symbol in $\cv_2$. Accordingly, we have $\vv_{2,i}=\zv$, which gives $\vv_{2,i}\rv_i^T=\zv\av^T+ \zv\Bm_{i}\bv^T$ and 
$$\cv_2=\{\zv\av^T+\zv\Bm_1\bv^T,\cdots,\zv\av^T+\zv\Bm_d\bv^T\}.$$
Now, the second systematic node chooses the repair vector $\vv_{1,0} = (\omega^0+\cdots+ \omega^{\alpha-1})^{-1} \zv$ and sends
$\vv_{1,0}\cv_2^T= \nu\vv_{1,0}\av^T+\zv\bv^T$ to the first systematic node. Here $\nu$ denotes the sum of $\alpha$ ones over $\F_q$, which depends on the characteristics of the field $\F_q$. Then, the first systematic node solves $d+1$
equations $\{\nu\vv_{1,0}\av^T+\zv\bv^T,\vv_{1,1}\av^T+\zv\bv^T, \cdots, \vv_{1,d}\av^T+\zv\bv^T\}$
in $d+1$ unknowns $\{a_1,\cdots,a_\alpha,\zv\bv^T\}$~\cite{LeScouarnec:Exact12}.
This follows if the matrix $$\Mm=[\nu\vv_{1,0},1; \vv_{1,1}, 1; \cdots; \vv_{1,d}, 1]$$ is invertible,
 when we represent the observed symbols at the first systematic node as $\Mm[a_1,\cdots,a_\alpha,\zv\bv^T]$. (Invertibility of this matrix is stated in~\cite{LeScouarnec:Exact12}. Our analysis shows that this holds if $q>n-1$ is a sufficiently large prime number such that the generator $w$ of $\FF_q$ used above satisfies 
$(\omega^{0} + \cdots + \omega^{\alpha-1})^2 w^{-(\alpha-1)} \notin \{0, \alpha^2\}$.
Details are omitted for brevity.)
The second systematic node can be repaired in a similar manner. It remains to show the secrecy of the file. Here, regardless of eavesdropped node being a systematic or a parity node, given the secure symbols, $\fv^s=\{s_1, \cdots, s_\alpha\}$, the eavesdropper
can obtain $\alpha$ equations in $\alpha$ unknowns $\rv=r_1, \cdots, r_\alpha$. This allows the eavesdropper to solve for $\rv$, and shows that $H(\rv|\fv^s,\ev)=0$, where the eavesdropper observes the content of the eavesdropped node, i.e., $\ev=\sv_{\Ec_1}$. We see that, at the eavesdropped node, the content of the stored data necessarily satisfies $H(\ev)=H(\sv_{\Ec_1})=\alpha$. Then, as the code satisfies both $H(\ev)\leq H(\rv)$ and $H(\rv|\fv^s,\ev)=0$, we obtain from Lemma~\ref{thm:SecrecyLemma} that $I(\fv^s;\ev)=I(s_1, \cdots, s_\alpha;\sv_{\Ec_1})=0$.

\textbf{Case $2$: $\mathcal{M}^{s} = \alpha-1$ when $(\ell_1,\ell_2)=(0,1)$}

We modify the above construction by considering the file given by
${\mathcal{M}}=\{a_1 \triangleq r_1, \cdots, a_\alpha \triangleq r_\alpha,
b_1 \triangleq r_1+s_1, \cdots, b_{\alpha-1} \triangleq r_{\alpha-1}+s_{\alpha-1},
b_\alpha\triangleq r_{\alpha+1}\}$.
The reconstruction and cooperative repair processes are the same as that of the previous case. We show that the secrecy constraint is satisfied here. The content of the eavesdropped node $\sv_{\Ec_2}$ is generated from the downloaded data $\dv_{\Ec_2}$.
Thus, we need to show $I(\fv^s;\ev)=0$ with
$\fv^s=\{s_1,\cdots,s_{\alpha-1}\}$ and $\ev=\dv_{\Ec_2}$.
Without loss generality, we assume that the eavesdropper observes the first systematic node. Considering the repair process described above, we have $\ev=\dv_{\Ec_2}=\{\vv_{1,0}\av^T+\zv\bv^T,\vv_{1,1}\av^T+\zv\bv^T, \cdots, \vv_{1,d}\av^T+\zv\bv^T\}$, from which we obtain that $H(\ev)\leq \alpha+1$. In addition, as the eavesdropper
can solve for $(\av,\zv\bv^T)$, it can solve for $\rv=\{r_1,\cdots,r_{\alpha+1}\}$ from the $\alpha+1$ number of equations in $(\av,\zv\bv^T)$, after canceling out the contribution of secure symbols $\fv^s=\{s_1,\cdots,s_{\alpha-1}\}$ from $\zv\bv^T$.
This shows that $H(\rv|\fv^s,\ev)=0$. Using this, together with
$H(\rv)=\alpha+1$ and Lemma~\ref{thm:SecrecyLemma}, we obtain
that $I(\fv^s;\ev)=I(s_1, \cdots, s_{\alpha-1};\dv_{\Ec_2})=0$.

\begin{figure}[t]
\centering
\includegraphics[width=0.7\columnwidth]{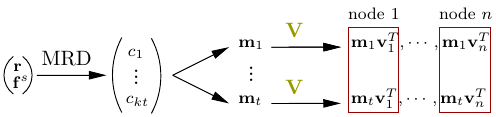}
\caption{MRD codeword is $\cv_{1:kt}$. 
MDS input vector is $\mv_j=(\cv_{(j-1)k+1},\ldots, \cv_{jk})$ for $j=1,\cdots,t$. 
Vandermonde matrix is represented with $\Vm$, columns of which are $\{\vv_1^T,\cdots,\vv_n^T\}$. The placement of the symbols on different nodes is shown on the right.}
\label{fig:MSCRdk}
\end{figure}

As a result of the above construction of secure MSCR codes with $k = t = 2$, we obtain the following result.
\begin{proposition}
The secrecy capacity at the MSCR point 
for a file size of $\Mc=k(d-k+t)\beta$ is given by
$\mathcal{M}^{s}=\alpha\beta$, if $(\ell_1,\ell_2)=(1,0)$ and $k=t=2$;
and by $\mathcal{M}^{s}=(\alpha-1)\beta$, if $(\ell_1,\ell_2)=(0,1)$ and $k=t=2$.
\end{proposition}


\subsection{Code construction for secure MSCR when $d=k$}

The above construction is limited to the $k=2$ case.
Here, we provide secure MSCR code when $d=k$, and hence
allowing $k>2$. (Note that as $d\geq k > \ell_1+\ell_2$,
we necessarily have $\ell_1+\ell_2 < d=k$ here.) Here, we apply a two-stage 
encoding, where we utilize an MRD code for pre-coding to achieve secrecy.

Consider ${\mathcal{M}}=k(d-k+t)=kt$, $\beta=\beta'=1$, $\alpha=d-k+t=t$,
$\mathcal{M}^{s}=(k-\ell_1-\ell_2)(t-\ell_2) = kt-(\ell_1+\ell_2)t-\ell_2(k-\ell_1-\ell_2)$, and $n\geq d+t$. (The general case of $\beta>1$ can be obtained by a parallel utilization of independent codes with $\beta=1$.)
We encode the data using the linearized polynomial
$f(g)=\sum\limits_{i=0}^{{\mathcal{M}}-1} u_i g^{q^i}$. 
(This is the Gabidulin construction of MRD codes summarized in Section~\ref{sec:MRD}.)
The coefficients of $f(g)$ is chosen by ${\mathcal{M}}-\mathcal{M}^{s}$ number of random symbols denoted by $\rv$ and $\mathcal{M}^{s}$ data symbols denoted by $\fv^s$.
The function $f(g)$ is evaluated at ${\mathcal{M}}$ points in $\F_{q^m}$, 
$\{g_1,\ldots, g_{\mathcal{M}}\}$, that are linearly independent over $\mathbb{F}_q$. (Here, the data and random symbols
belong to $\FF_{q^m}$ with $m\geq {\mathcal{M}}$.) We denote these evaluations as $c_i=f(g_i)$ for $i=1,\cdots,{\mathcal{M}}=kt$.
We consider the code provided in~\cite{Shum:Cooperative11}
for the secrecy setting here. We place these ${\mathcal{M}}=kt$ symbols
into vectors $\mv_1,\cdots,\mv_t$, each having $k$
symbols. We encode these vectors with a Vandermonde
matrix of size $k\times n$, whose columns are represented
as $\vv_i^T$ for $i=1,\cdots,n$. We store
$\{\mv_1 \vv_i^T,\ldots, \mv_t\vv_i^T\}$ at node $i$. (See Fig.~\ref{fig:MSCRdk}.)

Data collector $\rm{DC}$,
by contacting any $k$ nodes, can obtain $k$
equations for each of $\mv_j$, and solve them
to obtain $c_i$ for $i=1,\cdots,{\mathcal{M}}=kt$. It can then
obtain the secure data symbols by performing interpolation for the underlying linearized polynomial $f(g)$~\cite{Gabidulin:Theory85}. Next, we briefly describe the cooperative node repair process for the codes under consideration. For node repair, consider that node $j_l \in \{j_1, j_2,\ldots, j_t\}$ contacts
$d=k$ live nodes, referred to as $\{ i_1, i_2,\ldots, i_k\}$.
It downloads $\mv_l\vv_{i_r}^T$ from live node
$i_r$ for $r=1,\cdots,k$. Node $j_l$ then obtains $\mv_l$ by solving these $k$ equations.
It stores $\mv_l\vv_{j_l}^T$, and sends $\mv_l\vv_{j_{l'}}^T$ to node $j_{l'}\in \{j_1, j_2,\ldots, j_t\} ,  l' \neq l$,
i.e., the remaining nodes under repair. Each node $j_l \in\{j_1, j_2,\ldots, j_t\}$ repeats this procedure. As a result, node $j_l$ recovers its $\mv_{l'}\vv_{j_l}^T$ for $l' \in[1:t], l'\neq l$ by downloading a symbol from each node under repair.

We show the secrecy constraint has met assuming 
$\ell_2\leq t$ here. (Otherwise, this construction can not achieve
a positive secure file size as an eavesdropper obtains all $\mv_{1:t}$ symbols from download-eavesdropped nodes.)
We note that an eavesdropper obtain $\ell_2 k$ equations from the $d=k$ live nodes by observing repair of $\ell_2$ nodes (the observed symbols reveal $\ell_2$ number of $\mv_j$s), and an additional $\ell_2(t-\ell_2)=\ell_2(\alpha-\ell_2)$ symbols from the remaining nodes under repair. Besides this, the eavesdropper gets $\ell_1\alpha$ number of symbols from the content stored on $\ell_1$ storage-eavesdropped nodes. However, $\ell_1 \ell_2$ of these symbols
are linearly dependent to the ones downloaded by nodes in $\Ec_2$  (as the nodes in $\Ec_2$ recover $\ell_2$ number of $\mv_j$s during node repair). (See Fig.~\ref{fig:MSCRdkEve}.)
Therefore, using the given polynomial
and the secure data of length $\mathcal{M}^{s}$, the eavesdropper
can solve for the random symbols using
these $\ell_2(k+\alpha-\ell_2)+\ell_1(\alpha-\ell_2)=\ell_2(k+t-\ell_2)+\ell_1 (t-\ell_2)
=(k-\ell_1-\ell_2)\ell_2+(\ell_1+\ell_2)t={\mathcal{M}}-\mathcal{M}^{s}$
linearly independent evaluations of the polynomial $f(g)$. This implies that we have $H(\rv|\fv^s,\ev)=0$, where
$\ev$ denotes the observations of the eavesdropper associated with $\Ec_1$ and $\Ec_2$. This construction also satisfies
$H(\ev)=\ell_2k+\ell_2(\alpha-\ell_2)+\ell_1(\alpha-\ell_2)=H(\rv)$
as argued above, and it follows from Lemma~\ref{thm:SecrecyLemma}
that we have $I(\fv^s;\ev)=0$.
Thus, the proposed coding scheme achieves the secure file size of 
$kt-(\ell_1+\ell_2)t-\ell_2(k-\ell_1-\ell_2)$ when $\ell_2\leq t$; consequently, we have the following. 
\begin{proposition}
The secure file size of $\mathcal{M}^{s}
=(k-\ell_1-\ell_2)[t-\ell_2]^+\beta$ is achievable at the
MSCR point for a file size of ${\mathcal{M}}=k(d-k+t)\beta$
when $d=k$.
\end{proposition}
Note that this achieves the secrecy capacity when
$\ell_2\leq 1$ for any $\ell_1$ as can be observed from
the bound given by Proposition~\ref{thm:MSCRbound}.

\begin{figure*}[t]
\centering
\includegraphics[width=0.9\textwidth]{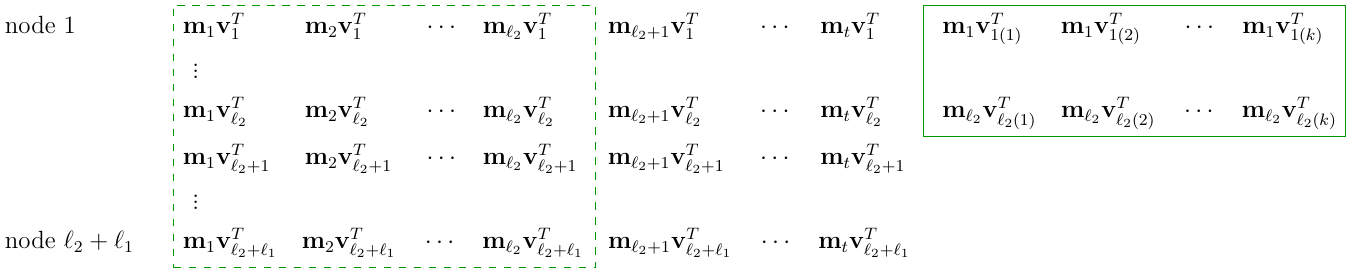}
\caption{The observed symbols at the eavesdropper. Without loss of generality we assume 
first $\ell_2$ nodes belong to $\Ec_2$ and the following $\ell_1$ nodes belong to
$\Ec_1$. We denote the symbols downloaded at $i\in\Ec_2$ from live node $x_{i(j)}$
as $\mv_i\vv_{i(j)}^T$ for $j=[1:d]$, where $i(j)$ denotes the $j${th} contacted live 
node for repair of node $i$. For the nodes in $\Ec_2$, we indicate these downloaded symbols on the right 
hand side of the figure, within the box with solid lines. (On the left hand side, we have the stored content of each node.)
From downloaded symbols to $\Ec_2$ nodes, the eavesdropper observes
$\dv_{\Ec_2}=\{\mv_{i}\vv_{i(j)}^T: i=1,\cdots,\ell_2; j=1,\cdots,k\}$. (Here, $d=k$.)
The symbols in the first $\ell_2$ columns, i.e., $\{\mv_{1}\vv_{1}^T,\ldots,\mv_{1}\vv_{\ell_2+\ell_1}^T,\cdots,\mv_{\ell_2}\vv_{1}^T,\ldots,\mv_{\ell_2}\vv_{\ell_2+\ell_1}^T\}$,
are functions of the symbols in $\dv_{\Ec_2}$ due to the construction (i.e., MDS coding).
These $\ell_2k$ symbols in $\dv_{\Ec_2}$ in addition to the $(\ell_1+\ell_2)(t-\ell_2)$ symbols 
located in the middle, i.e., $\{\mv_{\ell_2+1}\vv_{1}^T,\ldots,\mv_{\ell_2+1}\vv_{\ell_2+\ell_1}^T,\cdots,\mv_{t}\vv_{1}^T,\ldots,\mv_{t}\vv_{\ell_2+\ell_1}^T\}$,
correspond to $\ell_2k+(\ell_1+\ell_2)(t-\ell_2)$ linearly independent evaluations of the linearized polynomial $f(g)$.}
\label{fig:MSCRdkEve}
\end{figure*}


\section{Conclusion}
\label{sec:con}

Distributed storage systems (DSS) store data on multiple nodes. These systems not
only require resilience against node failures, but also have to satisfy security 
constraints and to perform multiple node repairs.
Regenerating codes proposed for DSS address the node failure resilience while efficiently
trading off storage vs. repair bandwidth. In this paper, we considered
secure cooperative regenerating codes for DSS against passive eavesdropping attacks. 
The eavesdropper model analyzed in this paper belongs to the class of passive attack models,
where the eavesdroppers observe the content of the nodes in the system.
Accordingly, we considered an $(\ell_1,\ell_2)$-eavesdropper, where the stored content at
any $\ell_1$ nodes, and the downloaded content at any $\ell_2$ nodes are leaked
to the eavesdropper. With such an eavesdropper model, we studied the security 
for the multiple repair scenario, in particular secure cooperative regenerating codes. 
For the minimum bandwidth cooperative regenerating (MBCR) point, we established a bound on the secrecy 
capacity, and by modifying the existing coding schemes in the literature,
devised new codes achieving the secrecy capacity. For the minimum storage
cooperative regenerating (MSCR) point, on the other hand, we proposed an upper
bound and lower bounds on the secure file size, which match under
special cases. The results show that it is possible to design regenerating
codes that not only efficiently trades storage for repair bandwidth, but are also
resilient against security attacks in a cooperative repair scenario.
Finally, as evident from some of our secrecy-achieving constructions, 
we would like to emphasize the role that the maximum rank distance (MRD) 
codes can take in secrecy problems. In particular, we have utilized the 
Gabidulin construction~\cite{Gabidulin:Theory85} of MRD codes and properties of 
linearized polynomials in obtaining some of the results. Similar properties 
of such codes have been utilized to achieve secrecy in earlier works
\cite{Gabidulin:Ideals91,Gibson:security96,Koetter:Coding07,Silva:rank08},
and they proved their potential again here as an essential component 
for achieving secrecy in DSS. (See also~\cite{Silberstein:Error12} for utilization
of these codes in active eavesdropper settings, and~\cite{Rawat:Optimal12} for constructing locally repairable codes with and without secrecy constraints.)

We list some avenues for further research here. The secrecy
capacity of MSCR codes remains as an open problem, as we have established
the optimal codes under some parameter settings. To attempt solving this problem,
codes for MSCR without security constraints have to be further investigated.
One can also consider cooperative repair in a DSS having locally
repairable structure. As other distributed systems, DSS may exhibit
simultaneous node failures that need to be recovered with local connections.
According to our best knowledge, this setting has not been studied (even
without security constraints). Our ongoing efforts are on the design
of coding schemes for DSS satisfying these properties.


\section*{Acknowledgement}
The authors would like to thank the reviewers for their insightful comments and suggestions.


\appendices

\section{Proof of Lemma~\ref{thm:SecrecyLemma}}
\label{app:SecrecyLemma}
\begin{proof}
The proof follows from the classical techniques given
by~\cite{Wyner:Wire-tap75}, where instead of
$0$-leakage, $\epsilon$-leakage rate is considered.
The application of this technique in DSS is
considered in~\cite{Shah:Information11}, as summarized below.
We have
\begin{eqnarray*}
I(\fv^s;\ev) &=& H(\ev)-H(\ev|\fv^s) \\
&\stackrel{(a)}{\leq}&H(\ev)-H(\ev|\fv^s) + H(\ev|\fv^s,\rv) \\
&\stackrel{(b)}{\leq}&H(\rv)-I(\ev;\rv|\fv^s) \\
&\stackrel{(c)}{=}&H(\rv|\fv^s,\ev) \\
&\stackrel{(d)}{=}&0,
\end{eqnarray*}
where (a) follows by non-negativity of $H(\ev|\fv^s,\rv)$,
(b) is the condition $H(\ev)\leq H(\rv)$,
(c) is due to $H(\rv|\fv^s)=H(\rv)$ as $\rv$ and $\fv^s$ are
independent, (d) is the condition $H(\rv|\fv^s,\ev)=0$.
\begin{remark}
If the eavesdropper has a vanishing probability of error
in decoding $\rv$ given $\ev$ and $\fv^s$, then, by Fano's
inequality, one can write $H(\rv|\fv^s,\ev)\leq |\rv|\epsilon$,
and, by following the above steps, can show the bound
$I(\fv^s;\ev)\leq |\rv|\epsilon$, where $|\rv|$ is the
number of random bits, and $\epsilon$ can be made small
if the probability of error is vanishing.
This shows that the leakage rate
$I(\fv^s;\ev)/|\ev|$ is vanishing.
(See, e.g.,~\cite{Wyner:Wire-tap75}.)
\end{remark}
\end{proof}


\section{Proof or Lemma~\ref{thm:SecureCutLemma}}
\label{app:SecureFileSizeBound}

We summarize the steps given in~\cite{Rawat:Optimal12}.

\begin{eqnarray*}
\Mc^s&=&H(\fv^s) \\
&\stackrel{(a)}{=}& H(\fv^s) - I(\fv^s;\sv_{\Ec_1'},\dv_{\Ec_2'})\\
&=& H(\fv^s|\sv_{\Ec_1'},\dv_{\Ec_2'})\\
&\stackrel{(b)}{=}& H(\fv^s|\sv_{\Ec_1'},\dv_{\Ec_2'}) - H(\fv^s|\sv_\Kc)  \\
&\stackrel{(c)}{=}& H(\fv^s|\sv_{\Ec_1'},\dv_{\Ec_2'}) - H(\fv^s|\sv_{\Ec_1'}, \dv_{\Ec_2'},\sv_{\Kc})  \\
&\leq& I(\fv^s;\sv_{\Kc}|\sv_{\Ec_1'},\dv_{\Ec_2'})\\
&\leq& H(\sv_{\Kc}|\sv_{\Ec_1'},\dv_{\Ec_2'})\\
&=& \sum\limits_{j=1}^k H(\sv_j|\sv_{1},\cdots,\sv_{j-1},\sv_{\Ec_1'},\dv_{\Ec_2'}),
\end{eqnarray*}
where (a) follows by the security constraint, i.e., $0 \leq I(\fv^s;\sv_{\Ec_1'},\dv_{\Ec_2'})\leq I(\fv^s;\sv_{\Ec_1},\dv_{\Ec_2})=0$, (b) is due to the data construction property, i.e., $H(\fv^s|\sv_\Kc)=0$, (c) is due to $0 \leq H( \fv^s|\sv_{\Ec'_1}, \dv_{\Ec'_2}, \sv_{\Kc}) \leq H( \fv^s|\sv_\Kc) = 0$.


\section{NRBW values for MBCR point in DSS}\label{app:CoopTable}

The parameters of Proposition~\ref{thm:MBCRbound} are given in the following tables.
$\ell_1=0$ case corresponds to the systems without security constraints.
$t=1$ case corresponds to non-cooperative case.
Red (green) font highlights cases with greater (respectively, smaller)
cooperative NRBW ($\gamma/\mathcal{M}^{s}$) compared to that for $t=1$.
We observed that the same trend continues for higher $n$ values.

\small
\begin{table}[htbp]
\caption{NRBW for $n=4,5$, $d\geq k$, $d+t=n$.}
\label{tab:Coop}
\centering
\begin{tabular}{|cccccccccc|}\hline
$n$ & $k$ & $l$ & $t$ & $d$ & $\beta/\mathcal{M}^{s}$ & $\beta'/\mathcal{M}^{s}$ & $\gamma/\mathcal{M}^{s}$ & ${\mathcal{M}}$ & $\mathcal{M}^{s}$ \\ \hline
  4 &   2 &   0 &   1 &   3 & 0.2000 & 0.1000 & 0.6000 &  10 &  10\\ \hline

  4 &   2 &   0 &   2 &   2 & 0.2500 & 0.1250 & \textbf{\textcolor{red}{0.6250}} &   8 &   8\\ \hline

  4 &   2 &   1 &   1 &   3 & 0.5000 & 0.2500 & 1.5000 &  10 &   4\\ \hline

  4 &   2 &   1 &   2 &   2 & 0.6667 & 0.3333 & \textbf{\textcolor{red}{1.6667}} &   8 &   3\\ \hline

  4 &   3 &   0 &   1 &   3 & 0.1667 & 0.0833 & 0.5000 &  12 &  12\\ \hline

  4 &   3 &   1 &   1 &   3 & 0.3333 & 0.1667 & 1.0000 &  12 &   6\\ \hline

  4 &   3 &   2 &   1 &   3 & 1.0000 & 0.5000 & 3.0000 &  12 &   2\\ \hline

  5 &   2 &   0 &   1 &   4 & 0.1429 & 0.0714 & 0.5714 &  14 &  14\\ \hline

  5 &   2 &   0 &   2 &   3 & 0.1667 & 0.0833 & \textbf{\textcolor{red}{0.5833}} &  12 &  12\\ \hline

  5 &   2 &   0 &   3 &   2 & 0.2000 & 0.1000 & \textbf{\textcolor{red}{0.6000}} &  10 &  10\\ \hline

  5 &   2 &   1 &   1 &   4 & 0.3333 & 0.1667 & 1.3333 &  14 &   6\\ \hline

  5 &   2 &   1 &   2 &   3 & 0.4000 & 0.2000 & \textbf{\textcolor{red}{1.4000}} &  12 &   5\\ \hline

  5 &   2 &   1 &   3 &   2 & 0.5000 & 0.2500 & \textbf{\textcolor{red}{1.5000}} &  10 &   4\\ \hline

  5 &   3 &   0 &   1 &   4 & 0.1111 & 0.0556 & 0.4444 &  18 &  18\\ \hline

  5 &   3 &   0 &   2 &   3 & 0.1333 & 0.0667 & \textbf{\textcolor{red}{0.4667}} &  15 &  15\\ \hline

  5 &   3 &   1 &   1 &   4 & 0.2000 & 0.1000 & 0.8000 &  18 &  10\\ \hline

  5 &   3 &   1 &   2 &   3 & 0.2500 & 0.1250 & \textbf{\textcolor{red}{0.8750}} &  15 &   8\\ \hline

  5 &   3 &   2 &   1 &   4 & 0.5000 & 0.2500 & 2.0000 &  18 &   4\\ \hline

  5 &   3 &   2 &   2 &   3 & 0.6667 & 0.3333 & \textbf{\textcolor{red}{2.3333}} &  15 &   3\\ \hline

  5 &   4 &   0 &   1 &   4 & 0.1000 & 0.0500 & 0.4000 &  20 &  20\\ \hline

  5 &   4 &   1 &   1 &   4 & 0.1667 & 0.0833 & 0.6667 &  20 &  12\\ \hline

  5 &   4 &   2 &   1 &   4 & 0.3333 & 0.1667 & 1.3333 &  20 &   6\\ \hline

  5 &   4 &   3 &   1 &   4 & 1.0000 & 0.5000 & 4.0000 &  20 &   2\\ \hline

\end{tabular}
\end{table}

\small
\begin{table}[htbp]
\caption{NRBW for $n=4,5$, $d\geq k$, $d+t\leq n$.}
\label{tab:Coop2}
\centering
\begin{tabular}{|cccccccccc|}\hline
$n$ & $k$ & $l$ & $t$ & $d$ & $\beta/\mathcal{M}^{s}$ & $\beta'/\mathcal{M}^{s}$ & $\gamma/\mathcal{M}^{s}$ & ${\mathcal{M}}$ & $\mathcal{M}^{s}$ \\ \hline
  4 &   2 &   0 &   1 &   3 & 0.2000 & 0.1000 & 0.6000 &  10 &  10\\ \hline

  4 &   2 &   0 &   1 &   2 & 0.3333 & 0.1667 & 0.6667 &   6 &   6\\ \hline

  4 &   2 &   0 &   2 &   2 & 0.2500 & 0.1250 & \textbf{\textcolor{green}{0.6250}} &   8 &   8 \\ \hline

  4 &   2 &   1 &   1 &   3 & 0.5000 & 0.2500 & 1.5000 &  10 &   4\\ \hline

  4 &   2 &   1 &   1 &   2 & 1.0000 & 0.5000 & 2.0000 &   6 &   2\\ \hline

  4 &   2 &   1 &   2 &   2 & 0.6667 & 0.3333 & \textbf{\textcolor{green}{1.6667}} &   8 &   3 \\ \hline

  4 &   3 &   0 &   1 &   3 & 0.1667 & 0.0833 & 0.5000 &  12 &  12\\ \hline

  4 &   3 &   1 &   1 &   3 & 0.3333 & 0.1667 & 1.0000 &  12 &   6\\ \hline

  4 &   3 &   2 &   1 &   3 & 1.0000 & 0.5000 & 3.0000 &  12 &   2\\ \hline

  5 &   2 &   0 &   1 &   4 & 0.1429 & 0.0714 & 0.5714 &  14 &  14\\ \hline

  5 &   2 &   0 &   1 &   3 & 0.2000 & 0.1000 & 0.6000 &  10 &  10\\ \hline

  5 &   2 &   0 &   2 &   3 & 0.1667 & 0.0833 & \textbf{\textcolor{green}{0.5833}} &  12 &  12 \\ \hline

  5 &   2 &   0 &   1 &   2 & 0.3333 & 0.1667 & 0.6667 &   6 &   6\\ \hline

  5 &   2 &   0 &   2 &   2 & 0.2500 & 0.1250 & \textbf{\textcolor{green}{0.6250}} &   8 &   8 \\ \hline

  5 &   2 &   0 &   3 &   2 & 0.2000 & 0.1000 & \textbf{\textcolor{green}{0.6000}} &  10 &  10 \\ \hline

  5 &   2 &   1 &   1 &   4 & 0.3333 & 0.1667 & 1.3333 &  14 &   6\\ \hline

  5 &   2 &   1 &   1 &   3 & 0.5000 & 0.2500 & 1.5000 &  10 &   4\\ \hline

  5 &   2 &   1 &   2 &   3 & 0.4000 & 0.2000 & \textbf{\textcolor{green}{1.4000}} &  12 &   5 \\ \hline

  5 &   2 &   1 &   1 &   2 & 1.0000 & 0.5000 & 2.0000 &   6 &   2\\ \hline

  5 &   2 &   1 &   2 &   2 & 0.6667 & 0.3333 & \textbf{\textcolor{green}{1.6667}} &   8 &   3 \\ \hline

  5 &   2 &   1 &   3 &   2 & 0.5000 & 0.2500 & \textbf{\textcolor{green}{1.5000}} &  10 &   4 \\ \hline

  5 &   3 &   0 &   1 &   4 & 0.1111 & 0.0556 & 0.4444 &  18 &  18\\ \hline

  5 &   3 &   0 &   1 &   3 & 0.1667 & 0.0833 & 0.5000 &  12 &  12\\ \hline

  5 &   3 &   0 &   2 &   3 & 0.1333 & 0.0667 & \textbf{\textcolor{green}{0.4667}} &  15 &  15 \\ \hline

  5 &   3 &   1 &   1 &   4 & 0.2000 & 0.1000 & 0.8000 &  18 &  10\\ \hline

  5 &   3 &   1 &   1 &   3 & 0.3333 & 0.1667 & 1.0000 &  12 &   6\\ \hline

  5 &   3 &   1 &   2 &   3 & 0.2500 & 0.1250 & \textbf{\textcolor{green}{0.8750}} &  15 &   8 \\ \hline

  5 &   3 &   2 &   1 &   4 & 0.5000 & 0.2500 & 2.0000 &  18 &   4\\ \hline

  5 &   3 &   2 &   1 &   3 & 1.0000 & 0.5000 & 3.0000 &  12 &   2\\ \hline

  5 &   3 &   2 &   2 &   3 & 0.6667 & 0.3333 & \textbf{\textcolor{green}{2.3333}} &  15 &   3 \\ \hline

  5 &   4 &   0 &   1 &   4 & 0.1000 & 0.0500 & 0.4000 &  20 &  20\\ \hline

  5 &   4 &   1 &   1 &   4 & 0.1667 & 0.0833 & 0.6667 &  20 &  12\\ \hline

  5 &   4 &   2 &   1 &   4 & 0.3333 & 0.1667 & 1.3333 &  20 &   6\\ \hline

  5 &   4 &   3 &   1 &   4 & 1.0000 & 0.5000 & 4.0000 &  20 &   2\\ \hline

\end{tabular}
\end{table}

\clearpage




\end{document}